\documentclass[onefignum,onetabnum]{siamart171218}

\usepackage{lipsum}
\usepackage{amsfonts}
\usepackage{graphicx}
\usepackage{epstopdf}
\usepackage{algorithmic}
\usepackage{array}
\usepackage{url}
\usepackage{amsmath,amssymb,amsbsy}
\usepackage{paralist}
\usepackage{xcolor}
\usepackage{color}
\usepackage{graphicx}
\graphicspath{{./figs/}}
\usepackage{algorithm}
\usepackage{algorithmic}
\usepackage{wrapfig}
\usepackage{caption}
\usepackage{subcaption}
\usepackage{bm}
\usepackage{fancyhdr}
\usepackage{cite}
\usepackage{cleveref}

\newtheorem{thm}{Theorem}
\newtheorem{lem}{Lemma}
\newtheorem{cor}{Corollary}
\newtheorem{prop}{Proposition}
\newtheorem{defi}{Definition}

\newcommand{\R}{\mathbb{R}}

\newcommand{\e}{\begin{equation}}
\newcommand{\ee}{\end{equation}}
\newcommand{\en}{\begin{equation*}}
\newcommand{\een}{\end{equation*}}
\newcommand{\eqn}{\begin{eqnarray}}
\newcommand{\eeqn}{\end{eqnarray}}
\newcommand{\bmat}{\begin{bmatrix}}
\newcommand{\emat}{\end{bmatrix}}

\renewcommand{\P}[1]{\Pr\left[#1\right]}
\newcommand{\E}{\operatorname{E}}

\newcommand{\vct}[1]{\boldsymbol{#1}}
\newcommand{\mtx}[1]{\boldsymbol{#1}}

\newcommand{\T}{\mathrm{T}}
\newcommand{\starT}{{\star\mathrm{T}}}

\newcommand{\rank}{\operatorname{rank}}

\newcommand{\Sign}{\operatorname{Sign}}

\newcommand{\dist}{\operatorname{dist}}

\newcommand{\set}[1]{\mathbb{#1}}

\DeclareMathOperator*{\minimize}{\text{minimize}}

\newcommand{\cmt}[1]{}

\newcommand{\calA}{\mathcal{A}}
\newcommand{\calB}{\mathcal{B}}

\newcommand{\calO}{\mathcal{O}}
\newcommand{\calP}{\mathcal{P}}

\newcommand{\calU}{\mathcal{U}}

\newcommand{\calW}{\mathcal{W}}
\newcommand{\calX}{\mathcal{X}}

\newcommand{\va}{\vct{a}}

\newcommand{\vd}{\vct{d}}

\newcommand{\vs}{\vct{s}}

\newcommand{\vw}{\vct{w}}
\newcommand{\vx}{\vct{x}}
\newcommand{\vy}{\vct{y}}
\newcommand{\vz}{\vct{z}}
\newcommand{\vzero}{\vct{0}}

\newcommand{\mA}{\mtx{A}}
\newcommand{\mB}{\mtx{B}}
\newcommand{\mC}{\mtx{C}}
\newcommand{\mD}{\mtx{D}}
\newcommand{\mE}{\mtx{E}}

\newcommand{\mP}{\mtx{P}}
\newcommand{\mQ}{\mtx{Q}}
\newcommand{\mR}{\mtx{R}}

\newcommand{\mT}{\mtx{T}}
\newcommand{\mU}{\mtx{U}}
\newcommand{\mV}{\mtx{V}}
\newcommand{\mW}{\mtx{W}}
\newcommand{\mX}{\mtx{X}}

\newcommand{\mDelta}{\mtx{\Delta}}

\newcommand{\mPhi}{\mtx{\Phi}}
\newcommand{\mPsi}{\mtx{\Psi}}
\newcommand{\mSigma}{\mtx{\Sigma}}

\newcommand{\mId}{{\bm I}}

\newcommand{\mzero}{{\bm 0}}

\newcommand{\mPi}{{\bf \Pi}}

\newcommand{\setI}{\set{I}}

\newcommand{\setS}{\set{S}}

\graphicspath{{./figs/}}

\newlength{\imgwidth}
\newcommand{\notexli}[1]{\textcolor{blue}{\bf [{\em Note: xli:} #1]}}

\newboolean{twoColVersion}
\setboolean{twoColVersion}{false}
\newcommand{\twoCol}[2]{\ifthenelse{\boolean{twoColVersion}} {#1} {#2} }

\usepackage[framemethod=tikz]{mdframed}

\ifpdf
  \DeclareGraphicsExtensions{.eps,.pdf,.png,.jpg}
\else
  \DeclareGraphicsExtensions{.eps}
\fi

\newsiamremark{remark}{Remark}
\newsiamremark{hypothesis}{Hypothesis}
\crefname{hypothesis}{Hypothesis}{Hypotheses}
\newsiamthm{claim}{Claim}

\headers{Nonconvex Robust Low-rank Matrix Recovery}{X. Li, Z. Zhu, A. M.-C. So, R. Vidal}

\title{Nonconvex Robust Low-rank Matrix Recovery\thanks{Submitted to the editors \today. The first and second authors contributed equally to this paper.
		\funding{Z.~Zhu and R.~Vidal were partially supported by NSF Grant 1704458. A.~M.-C.~So was partially supported by the Hong Kong Research Grants Council (RGC) General Research Fund (GRF) Project CUHK 14208117.}
}
}

\author{Xiao Li\thanks{Department of Electronic Engineering, The Chinese University of Hong Kong. 
  (\email{xli@ee.cuhk.edu.hk}, \url{http://www.ee.cuhk.edu.hk/\~xli/}).}
\and Zhihui Zhu \thanks{Center for Imaging Science, Mathematical Institute for Data Science, Johns Hopkins University. (\email{zzhu29@jhu.edu}, \url{http://cis.jhu.edu/\~zhihui/}; \email{rvidal@jhu.edu}, \url{http://cis.jhu.edu/\~rvidal/}).}
\and Anthony Man-Cho So \thanks{Department of Systems Engineering and Engineering Management, The Chinese University of Hong Kong. (\email{manchoso@se.cuhk.edu.hk}, \url{http://www.se.cuhk.edu.hk/\~manchoso/}).   }
\and Ren\'{e} Vidal\footnotemark[3]
}

\usepackage{amsopn}

\usepackage{changes} 

\begin{document}

\maketitle

\begin{abstract}
In this paper we study the problem of recovering a low-rank matrix from a number of random linear measurements that are corrupted by outliers taking arbitrary values. We consider a nonsmooth nonconvex formulation of the problem, in which we explicitly enforce the low-rank property of the solution by using a factored representation of the matrix variable and employ an $\ell_1$-loss function to robustify the solution against outliers. We show that even when a constant fraction (which can be up to almost half) of the measurements are arbitrarily corrupted, as long as certain measurement operators arising from the measurement model satisfy the so-called $\ell_1/\ell_2$-restricted isometry property, the ground-truth matrix can be exactly recovered from any global minimum of the resulting optimization problem. Furthermore, we show that the objective function of the optimization problem is sharp and weakly convex. Consequently, a subgradient Method (SubGM) with geometrically diminishing step sizes will converge linearly to the ground-truth matrix when suitably initialized. We demonstrate the efficacy of the SubGM for the nonconvex robust low-rank matrix recovery problem with various numerical experiments.
\end{abstract}



\section{Introduction}

Low-rank matrices are ubiquitous in computer vision~\cite{candes2011robust,haeffele2014structured}, machine learning~\cite{srebro2004maximum}, and signal processing~\cite{davenport2016overview} applications. One fundamental computational task is to recover a low-rank matrix $\mX^\star\in\R^{n_1\times n_2}$ from a small number of linear measurements
\e
\vy = \calA(\mX^\star),	
\label{eq:ms model}
\ee
where $\calA:\R^{n_1 \times n_2}\rightarrow \R^{m}$ is a known linear operator.  Such a task arises in quantum tomography~\cite{aaronson2007learnability}, face recognition~\cite{candes2011robust}, linear system identification~\cite{fazel2004rank}, collaborative filtering~\cite{candes2009exact}, etc. We refer the interested reader to~\cite{ZSY12,davenport2016overview} for more detailed discussions.

Although in many interesting scenarios the number of linear measurements $m$ is much smaller than $n_1n_2$, the low-rank property of $\mX^\star$ suggests that its degrees of freedom can also be much smaller than $n_1n_2$, thus making the task of recovering $\mX^\star$ possible. This has been demonstrated in, e.g.,~\cite{candes2009exact}, where a nuclear norm minimization appproach for recovering a low-rank matrix from random linear measurements is studied. Despite the strong theoretical guarantees of such approach (see also~\cite{gross11}), most existing methods for solving the nuclear norm minimization problem do not scale well with the problem size (i.e., $n_1$, $n_2$, and $m$). To overcome this computational bottleneck, one approach is to enforce the low-rank property explicitly by using a factored representation of the matrix variable in the optimization formulation. Such an approach
has already been explored in some early works on low-rank semidefinite programming (see, e.g.,~\cite{burer2003nonlinear,burer2005local} and the references therein) but has gained renewed interest lately in the study of low-rank matrix recovery problems. For the purpose of illustration, let us first consider the case where the ground-truth matrix $\mX^\star$ is symmetric positive semidefinite with rank $r$. Instead of optimizing, say, an $\ell_2$-loss function involving an $n\times n$ symmetric positive semidefinite matrix variable $\mX$ with either a constraint or a regularization term controlling the rank of $\mX$, we consider the factorization $\mX=\mU\mU^\T$ and optimize the loss function over the $n\times r$ matrix variable $\mU$:
\begin{align}
\minimize_{\mU\in\R^{n\times r}} \left\{ \xi(\mU) := \frac{1}{m}\|\vy - \calA(\mU\mU^\T)\|_2^2\right\}.
\label{eq:ms factorization}
\end{align}
There are two obvious advantages with the formulation~\eqref{eq:ms factorization}. First, the recovered matrix will automatically satisfy the rank and positive semidefinite constraints. Second, when the rank of the ground-truth matrix is small, the size of the variable $\mU$ can be much smaller than that of $\mX$. Although the quadratic nature of $\mU\mU^\T$ renders the objective function $\xi$ in \eqref{eq:ms factorization} nonconvex, recent advances in the analysis of the landscapes of structured  nonconvex functions allow one to show that when the linear measurement operator $\mathcal{A}$ satisfies certain restricted isometry property (RIP), local search algorithms (such as gradient descent) are guaranteed to find a global minimum of~\eqref{eq:ms factorization} and exactly recover the underlying low-rank matrix $\mX^\star$~\cite{sun2015guaranteed,bhojanapalli2016lowrankrecoveryl,ge2016matrix,park2016non,zhu2018global}. Moreover, it was shown in~\cite{zheng2015convergent,tu2015low} that~\eqref{eq:ms factorization} satisfies an error bound condition, indicating that simple gradient descent with an appropriate initialization will converge to a global minimum at a linear rate; see \cite{chi2018review} for a comprehensive review.

\subsection{Our Goal and Main Results}
In this paper, we consider the {\em robust low-rank matrix recovery problem}, in which the measurements are corrupted by {\em outliers}. Specifically, we assume that
\e
\vy = \calA(\mX^\star) + \vs^\star,
\label{eq:rms model}\ee
where  $\vs^\star\in\R^m$  is an outlier vector such that a small fraction of its entries (the outliers) have an arbitrary magnitude and the remaining entries are zero. Moreover, the set of nonzero entries is assumed to be unknown. Outliers are prevalent in the context of sensor calibration~\cite{li2017low} (because of sensor failure), face recognition~\cite{de2003framework} (due to self-shadowing, specularity, or saturations in brightness), video surveillance~\cite{li2004statistical} (where the foreground objects are modeled as outliers), etc.
\begin{figure}[!h]
	\begin{subfigure}{0.31\linewidth}
		\centerline{
			\includegraphics[height=1.2in]{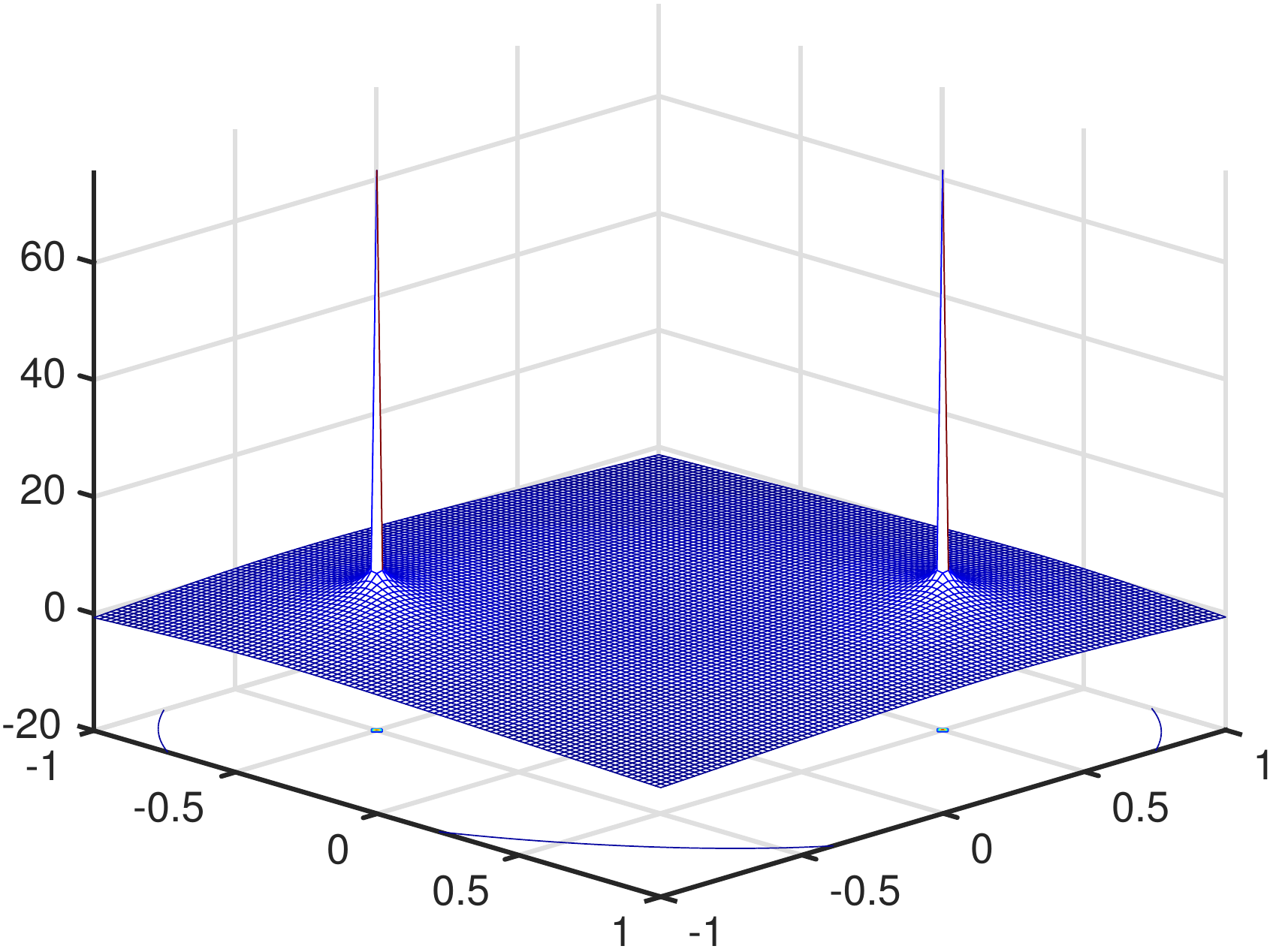}}
		\caption{No outliers }
	\end{subfigure}
	\begin{subfigure}{0.31\linewidth}
		\centerline{
			\includegraphics[height=1.2in]{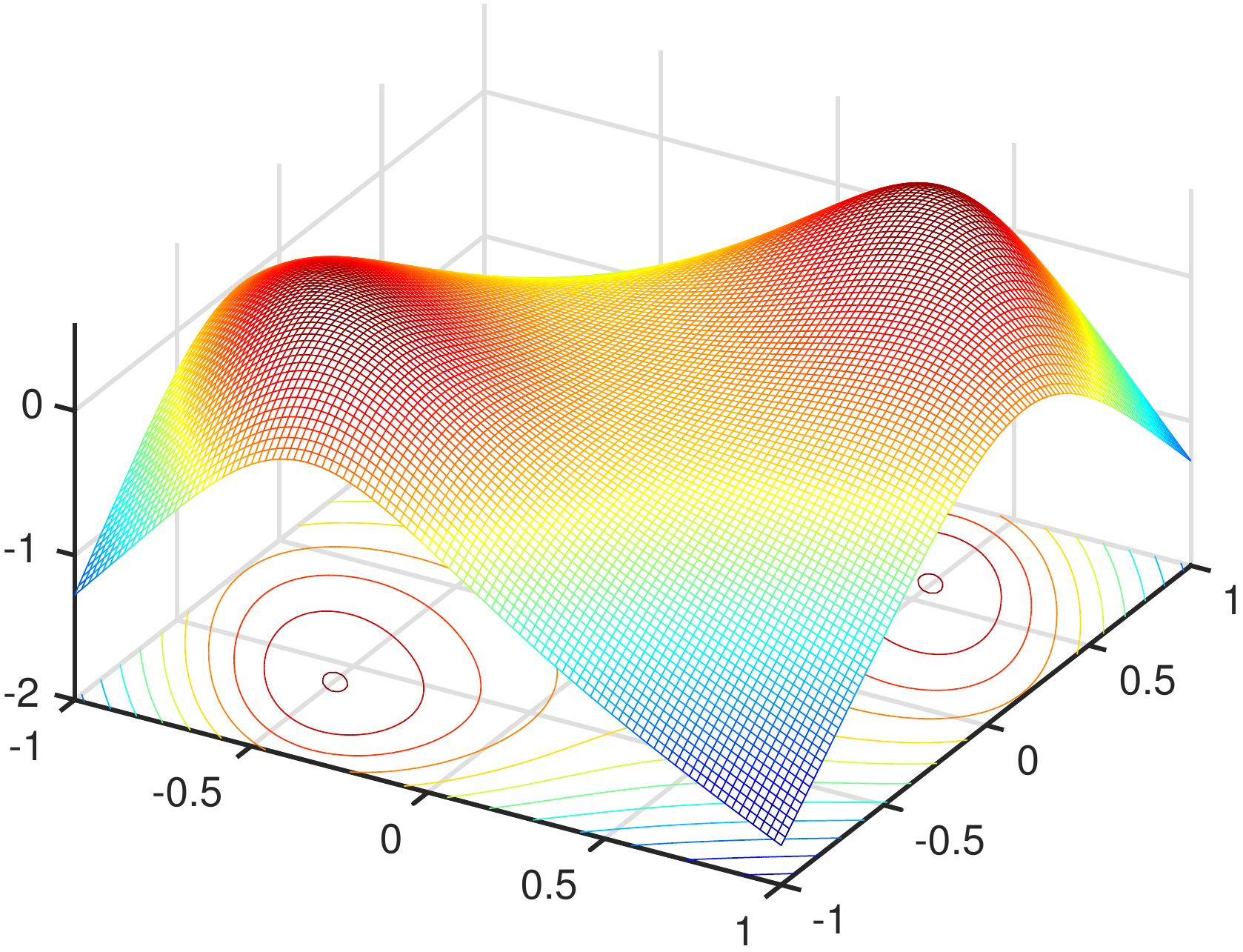}}
		\caption{$5\%$ outliers}
	\end{subfigure}
	\begin{subfigure}{0.31\linewidth}
		\centerline{
			\includegraphics[height=1.2in]{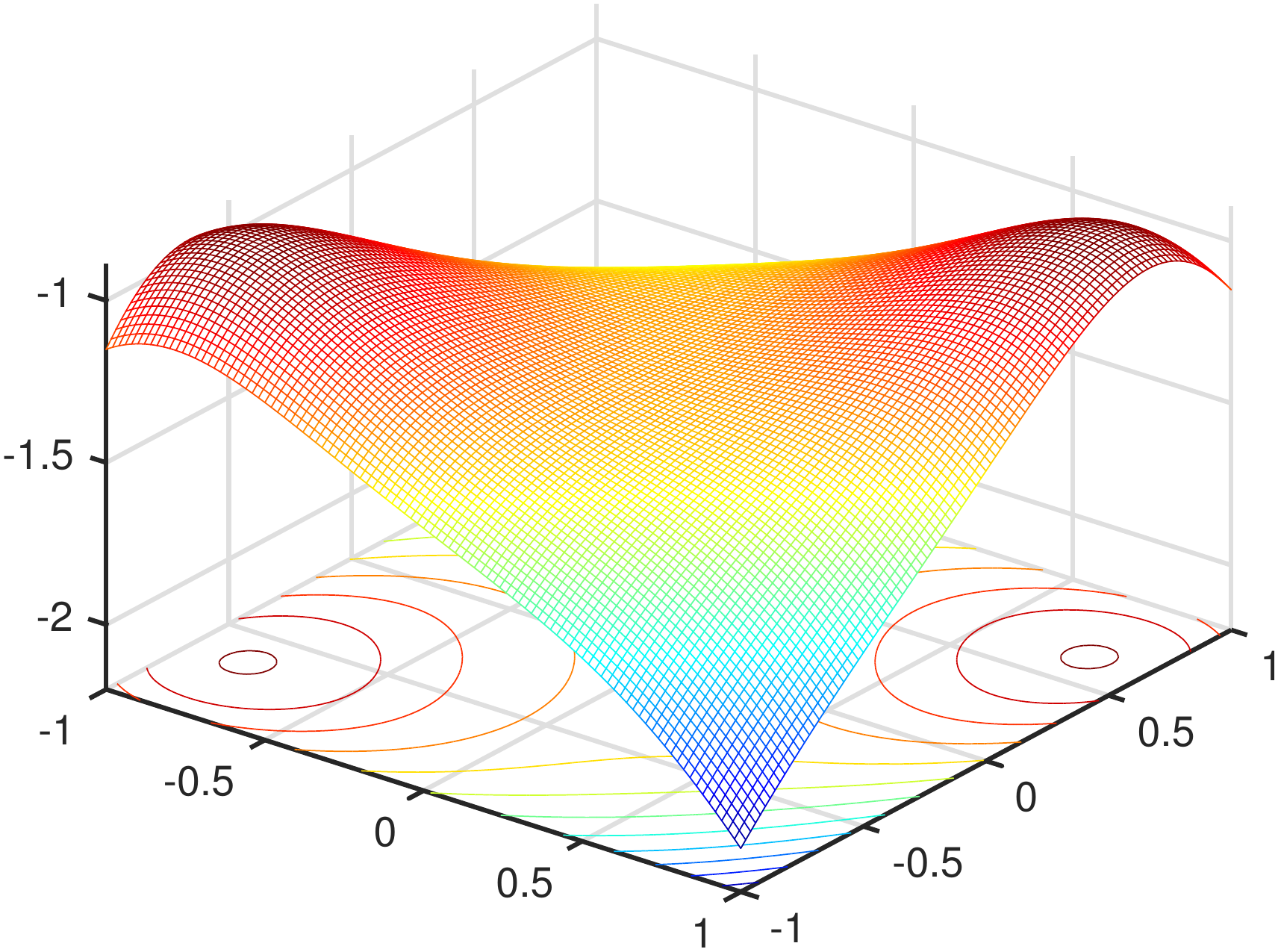}}
		\caption{$10\%$ outliers}
	\end{subfigure}
	\vfill
	\begin{subfigure}{0.31\linewidth}
		\centerline{
			\includegraphics[height=1.2in]{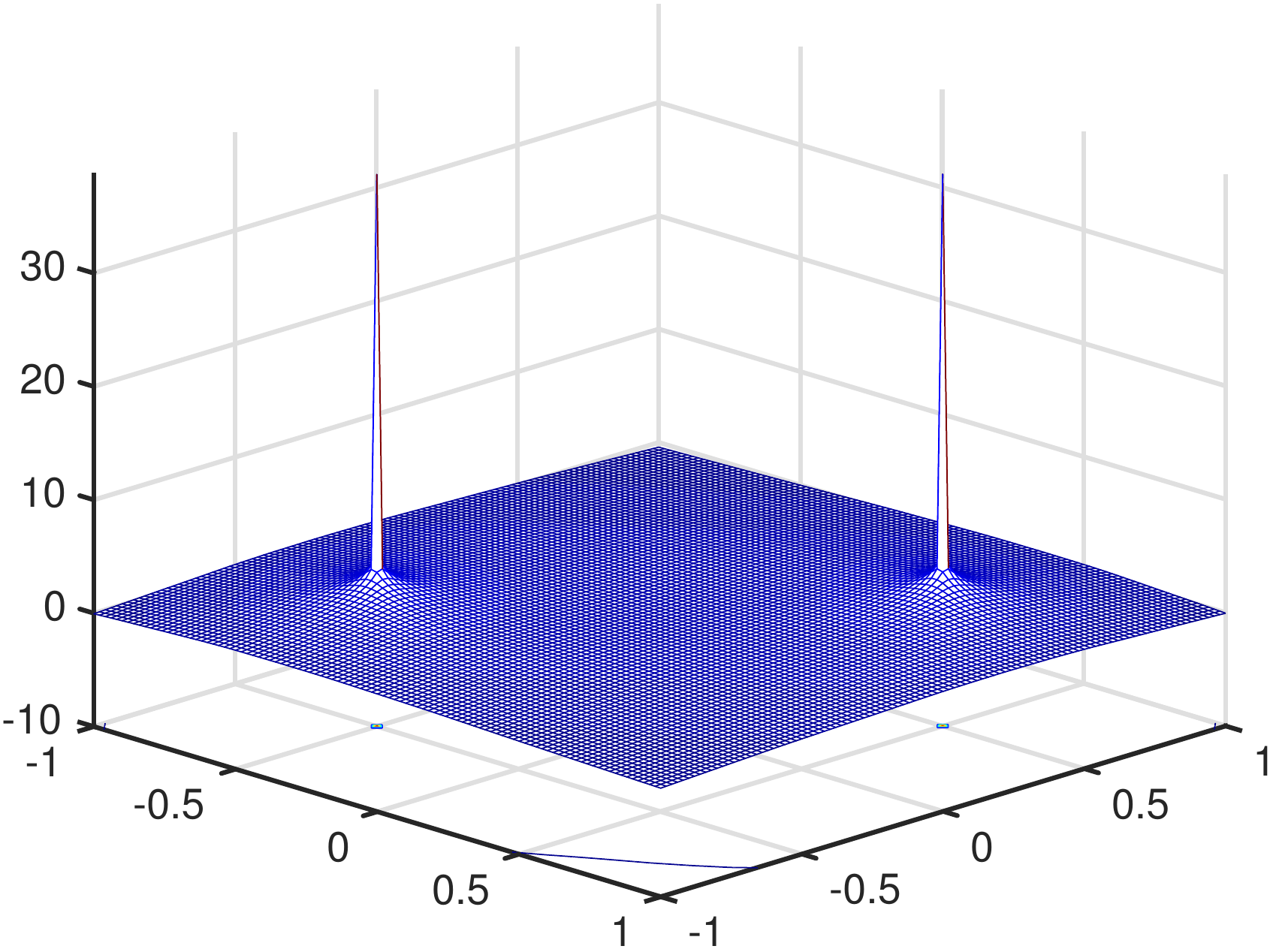}}
		\caption{No outliers }
	\end{subfigure}
	\begin{subfigure}{0.31\linewidth}
		\centerline{
			\includegraphics[height=1.2in]{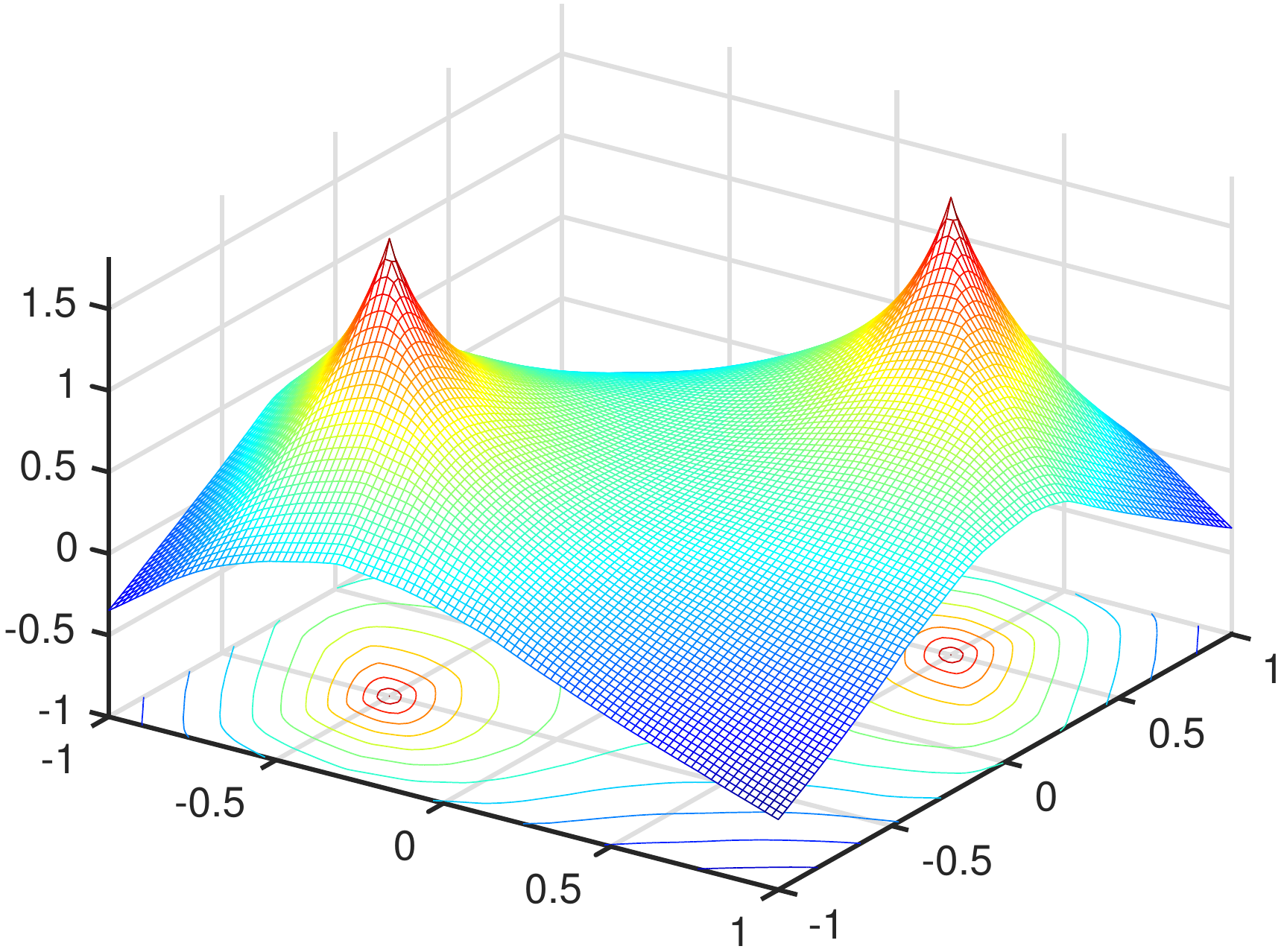}}
		\caption{$5\%$ outliers}
	\end{subfigure}
	\begin{subfigure}{0.31\linewidth}
		\centerline{
			\includegraphics[height=1.2in]{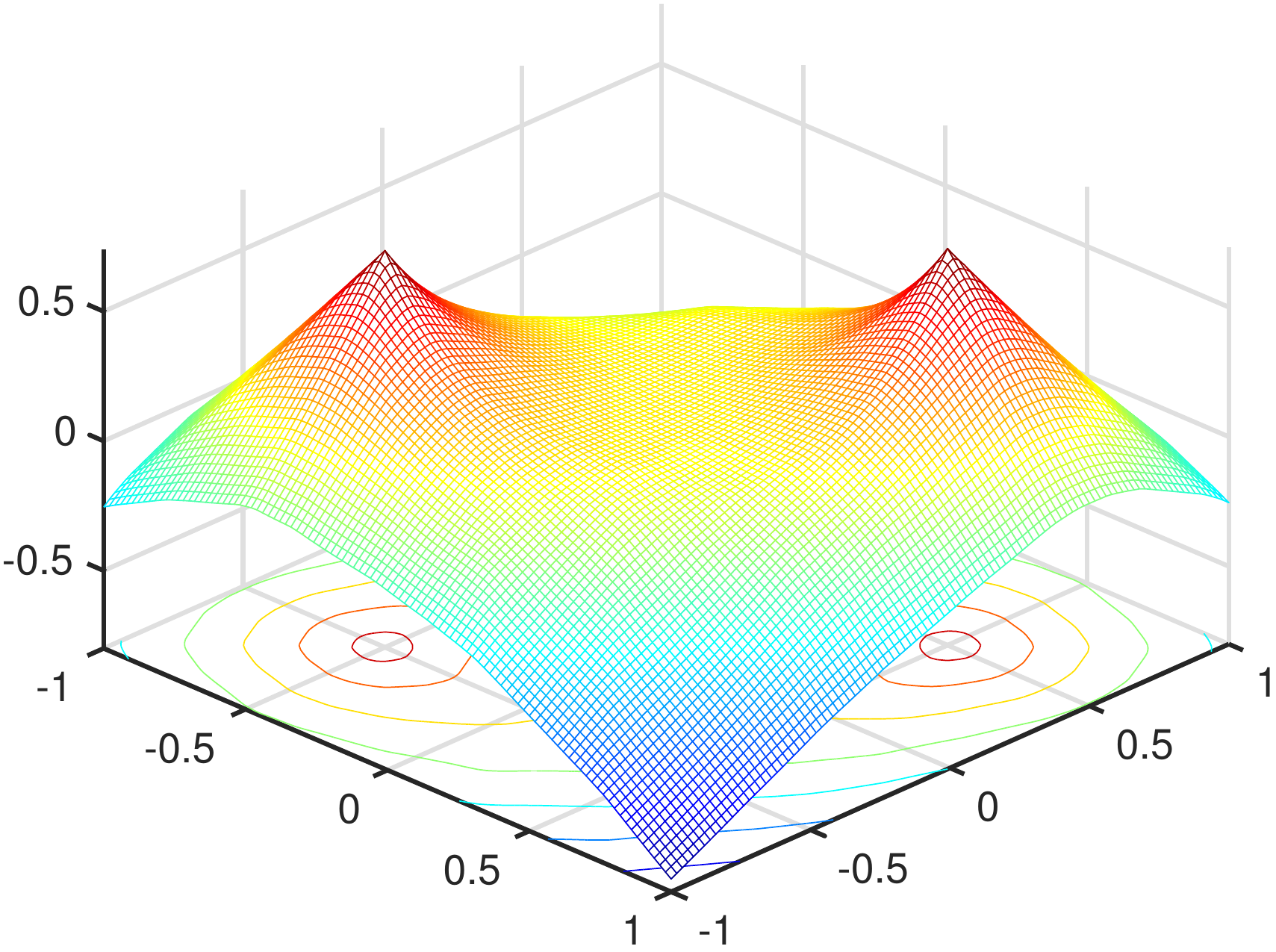}}
		\caption{$10\%$ outliers}
	\end{subfigure}
	\caption{\small Landscapes of the objective functions $\mU \mapsto \xi(\mU)= \frac{1}{m}\|\vy - \calA(\mU\mU^\T)\|_2^2$ (top row) and $\mU \mapsto f(\mU)= \frac{1}{m}\|\vy - \calA(\mU\mU^\T) \|_1$ (bottom row) for low-rank matrix recovery with different percentages of outliers in the measurement vector $\vy$ \eqref{eq:rms model}. Here, the ground-truth matrix $\mX^\star$ is given by $\mX^\star = \mU^\star \mU^{\star\T}$ with $\mU^\star = [0.5 \quad 0.5]^\T$ and $40$ measurements are taken to form $\vy$. For display purpose, we plot $-\log(\xi(\mU))$ and $-\log(f(\mU))$ instead of $\xi(\mU)$ and $f(\mU)$.}\label{fig:compare l1 and l2 loss}
\end{figure}
It is well known that the $\ell_2$-loss function is sensitive to outliers, thus rendering~\eqref{eq:ms factorization} ineffective for recovering the underlying low-rank matrix.  As illustrated in the top row of \Cref{fig:compare l1 and l2 loss}, the global minima of $\xi$ in \eqref{eq:ms factorization} are perturbed away from the underlying low-rank matrix because of the outliers, and a larger fraction of outliers leads to a larger perturbation.  By contrast, the $\ell_1$-loss function is more robust against outliers and has been widely utilized for outlier detection~\cite{candes2011robust,li2017low,josz2018theory}. This motivates us to adopt the $\ell_1$-loss function together with the factored representation of the matrix variable to tackle the robust low-rank matrix recovery problem:
\e
\minimize_{\mU\in\R^{n\times r}} \left\{ f(\mU):= \frac{1}{m}\|\vy - \calA(\mU\mU^\T) \|_1 \right\}.
\label{eq:rms factorization}
\ee
The robustness of the $\ell_1$-loss function against outliers can be seen from the bottom row of \Cref{fig:compare l1 and l2 loss}, where the global minima of \eqref{eq:rms factorization} correspond precisely to the underlying low-rank matrix $\mX^\star$ even in the presence of outliers. However, compared with~\eqref{eq:ms factorization}, the exact recovery property of~\eqref{eq:rms factorization} (i.e., when the global minima of~\eqref{eq:rms factorization} yield the ground-truth matrix $\mX^\star$) and the convergence behavior of local search algorithms for solving~\eqref{eq:rms factorization} are much less understood. This stems in part from the fact that~\eqref{eq:rms factorization} is a nonsmooth nonconvex optimization problem, but most of the algorithmic and analysis techniques developed in the recent literature on structured nonconvex optimization problems apply only to the smooth setting.


In view of the above discussion, we aim to (i) provide conditions in terms of the number of linear measurements $m$ and the fraction of outliers that can guarantee the exact recovery property of \eqref{eq:rms factorization} and (ii) design a first-order method to solve \eqref{eq:rms factorization} and establish guarantees on its convergence performance. To achieve (i), we utilize the notion of $\ell_1/\ell_2$-restricted isometry property ($\ell_1/\ell_2$-RIP), which has been introduced previously in the context of low-rank matrix recovery~\cite{ZHZ13,YS16} and covariance estimation~\cite{chen2015exact}. We show that if the fraction of outliers is slightly less than $\frac{1}{2}$, then as long as the measurement operator $\calA$ and its restriction $\calA_{\Omega^c}$ onto the complement of the support set $\Omega$ of the outlier vector $\vs^\star$ possess the $\ell_1/\ell_2$-RIP, any global minimum $\mU^\star$ of~\eqref{eq:rms factorization} must satisfy $\mU^\star \mU^{\star \T} = \mX^\star$. To tackle (ii), we propose to use a subgradient method (SubGM) to solve~\eqref{eq:rms factorization}. As a key step in our convergence analysis of the SubGM, we show that under the aforementioned setting for the fraction of outliers and the $\ell_1/\ell_2$-RIP of the operators $\calA$ and $\calA_{\Omega^c}$, the objective function $f$ in \eqref{eq:rms factorization} is \emph{sharp} (see \Cref{def:shaprness}) and \emph{weakly convex} (see \Cref{def:weak convexity}). Consequently, we can apply (a slight variant of) the analysis framework in~\cite{davis2018subgradient} to show that when initialized close to the set of global minima of~\eqref{eq:rms factorization}, the SubGM with geometrically diminishing step sizes will converge $R$-linearly to a global minimum. To the best of our knowledge, this is the first time an exact recovery condition (i.e., the $\ell_1/\ell_2$-RIP of $\calA$ and $\calA_{\Omega^c}$) for the optimization formulation~\eqref{eq:rms factorization} is shown to also imply its regularity (i.e., sharpness and weak convexity). We summarize the above results in the following theorem:

\begin{thm}[informal;  see \Cref{thm:rms} for the formal statement] 
Consider the measurement model~\eqref{eq:rms model}, where the ground-truth matrix $\mX^\star$ is symmetric positive semidefinite with rank $r$. Suppose that the fraction of outliers is less than half and both operators $\calA$ and $\calA_{\Omega^c}$ possess the $\ell_1/\ell_2$-RIP (see~\Cref{subsec:L1_2 RIP} and~\Cref{subsec:sharp}). Then, every global minimum of \eqref{eq:rms factorization} corresponds to the ground-truth matrix $\mX^\star$ and the objective function $f$ is sharp (see \Cref{def:shaprness}) and weakly convex (see \Cref{def:weak convexity}). Consequently, when applied to~\eqref{eq:rms factorization}, the SubGM with an appropriate initialization will converge to the ground-truth matrix $\mX^\star$ at a linear rate.
\label{thm:global optimality informal}
\end{thm}

Before we proceed, several remarks are in order. First, for various random measurement operators $\calA$, such as sub-Gaussian measurement operators and the quadratic measurement operators in~\cite{chen2015exact}, as long as the number of measurements is sufficiently large, the operators $\calA$ and $\calA_{\Omega^c}$ will possess the $\ell_1/\ell_2$-RIP with high probability. This is the case, for instance, when $\calA$ is a  Gaussian measurement operator with $m \gtrsim nr$ measurements.\footnote{See \Cref{subsec:not} for the meaning of the notation $\gtrsim$.} In particular, when combined with~\Cref{thm:global optimality informal}, we see that the low-rank matrix $\mX^\star$ in~\eqref{eq:rms model} can be recovered using an information-theoretically optimal number of measurements. Second, although at first glance \eqref{eq:rms factorization} seems to be more difficult to solve than \eqref{eq:ms factorization} because of nonsmoothness, \Cref{thm:global optimality informal} implies that~\eqref{eq:rms factorization} can be solved \emph{as efficiently as} its smooth counterpart \eqref{eq:ms factorization}, in the sense that both can be solved by first-order methods that have a linear convergence guarantee.

Although \Cref{thm:global optimality informal} is concerned with the setting where $\mX^\star$ is symmetric positive semidefinite, it can be extended to the general setting where $\mX^\star$ is a rank-$r$ $n_1\times n_2$ matrix. Specifically, by using the factorization $\mX = \mU\mV^\T$ with $\mU\in \R^{n_1\times r}$, $\mV \in \R^{n_2\times r}$ and utilizing the nonsmooth regularizer $\|\mU^\T\mU - \mV^\T\mV\|_F$ (or $\|\mU^\T\mU - \mV^\T\mV\|_1$) to account for the ambiguities in the factorization caused by invertible transformations, we formulate the general robust low-rank matrix recovery problem as follows:
\e
\minimize_{\mU\in\R^{n_1 \times r},\mV\in\R^{n_2 \times r}} \left\{ g(\mU,\mV) := \frac{1}{m}\|\vy - \calA(\mU\mV^\T) \|_1 + \lambda \|\mU^\T\mU - \mV^\T\mV\|_F \right\}.
\label{eq:nonsq-rms}
\ee
Here, $\lambda>0$ is a regularization parameter. We remark that the regularizer used in the above formulation is motivated by but different from that used in~\cite{tu2015low,park2016non,zhu2018global}. The latter, which is given by $\|\mU^\T\mU - \mV^\T\mV\|_F^2$, is smooth but is not as well suited for robustifying the solution against outliers. In~\Cref{sec:nonsymmetric} we show that all the results established for~\eqref{eq:rms factorization} in \Cref{thm:global optimality informal} carry over to~\eqref{eq:nonsq-rms} for any $\lambda>0$ (but the choice of $\lambda$ affects  the sharpness and weak convexity parameters; see the discussion after \Cref{lem:weak convex nonsymmetric}).

\subsection{Related Work}
By analyzing the optimization geometry, recent works~\cite{tu2015low,bhojanapalli2016lowrankrecoveryl,ge2016matrix,park2016non,li2016symmetry} have shown that many local search algorithms with either an appropriate initialization or a random initialization can provably solve the low-rank matrix recovery problem~\eqref{eq:ms factorization} when the measurement operator $\calA$ satisfies the RIP. In particular, gradient descent with an appropriate initialization is shown to converge to a global optimum at a linear rate \cite{tu2015low,zhu2017global}, while quadratic convergence is established for the cubic regularization method~\cite{yue2018quadratic}. Key to these results is certain error bound conditions, which elucidate the regularity properties of the underlying optimization problem. Recently, the above results have been extended to cover general smooth low-rank matrix optimization problems whose objective functions satisfy the restricted strong convexity and smoothness properties~\cite{zhu2018global,li2016,zhu2017global}.

For the robust low-rank matrix recovery problem, existing solution methods can be classified into two categories. The first is based on the convex approach~\cite{ke2005robust,candes2011robust,li2017low}. Although such approach enjoys strong statistical guarantees, it is computational expensive and thus not scalable to practical problems. The second category is based on the nonconvex approach. This includes the alternating minimization methods\cite{netrapalli2014non,yi2016fast,gu2016low,zhang2017unified}, which typically use projected gradient descent for low-rank matrix recovery and thresholding-based truncation for identification of outliers. However, these methods typically require performing an SVD in each iteration for projection onto the set of low-rank matrices. Recently, a median-truncated gradient descent method has been proposed in~\cite{li2017nonconvex} to tackle~\eqref{eq:ms factorization}, where the gradient is modified to alleviate the effect of outliers.  The median-truncated gradient descent is shown to have a local linear convergence rate~\cite{li2017nonconvex}, but such guarantee requires $m \gtrsim nr\log n$ measurements. Moreover, the maximum number of outliers that can be tolerated is not explicitly given. By contrast, our result only requires $m \gtrsim nr$ measurements (which matches the optimal information-theoretic bound) and explicitly bounds the fraction of outliers that can be present. We also note that a SubGM has been proposed in~\cite{li2017low} for solving~\eqref{eq:rms factorization} in the setting where $\mathcal{A}$ is a certain quadratic measurement operator. As reported in~\cite{li2017low}, the SubGM exhibits excellent empirical performance in terms of both computational efficiency and accuracy. In this paper, we provide a rigorous justification for the empirical success of the SubGM, thus answering a question that is left open in~\cite{li2017low}.

Finally, we remark that our work is closely related to the recent works~\cite{davis2017nonsmooth,davis2018subgradient,Zhu18DPCP,bai2018DL} on subgradient methods for nonsmooth nonconvex optimization. A projected subgradient method is proven to converge linearly for the robust subspace recovery problem~\cite{Zhu18DPCP} and sublinearly for orthonormal dictionary learning~\cite{bai2018DL}. It is shown in~\cite{davis2017nonsmooth,davis2018subgradient} that if the optimization problem at hand is sharp (see \Cref{def:shaprness}) and weakly convex (see \Cref{def:weak convexity}), various subgradient methods for solving it will converge at a linear rate. Currently, only a few applications are known to give rise to sharp and weakly convex optimization problems, such as robust phase retrieval~\cite{davis2017nonsmooth,duchi2017solving} and robust covariance estimation with quadratic sampling~\cite{davis2018subgradient}. Thus, our result expands the repertoire of optimization problems that are sharp and weakly convex and contributes to the growing literature on the geometry of structured nonsmooth nonconvex optimization problems.


\subsection{Notation} \label{subsec:not}
Let us introduce the notations used in this paper. Finite-dimensional vectors and matrices are indicated by bold characters. The symbols $\mId$ and $\mzero$ represent the identity matrix and zero matrix/vector, respectively. The set of $r\times r$ orthogonal matrices is denoted by $\calO_r:=\{\mR\in\R^{r\times r}:\mR^\T\mR = \mId\}$. The subdifferential of the absolute value function $|\cdot|$ is denoted by $\Sign$; i.e.,
\[
\Sign(a):=\left\{\begin{matrix}a/|a|, & a\neq 0,\\ [-1,1], & a= 0.\end{matrix}\right.
\]
We use $\Sign(\mA)$ to denote the matrix obtained by applying the Sign function to each element of the matrix $\mA$. Furthermore, we use $\|\mA\|_F$ to denote the Frobenius norm of the matrix $\mA$ and $\|\va\|$ to denote the $\ell_2$-norm of the vector $\va$. Finally, we use $x\lesssim y$ (resp.~$x \gtrsim y$) to indicate that $x \leq cy$ (resp.~$x \ge cy$) for some universal constant $c>0$.

\section{Problem Setup and Preliminaries} \label{sec:preliminaries}

Consider the general optimization problem
\e
\inf_{\vx\in \R^n} h(\vx),
\label{eq:min h}
\ee
where $h:\R^n\rightarrow \R$ is a lower semi-continuous, possibly nonsmooth and nonconvex, function.  Let $h^\star$ denote the optimal value of~\eqref{eq:min h} and
\[ \calX: = \{\vz\in \R^n: h(\vz)\leq h(\vx), \ \forall \vx \in \R^n\} \]
denote the set of global minima of $h$.  We assume that $\calX \not= \emptyset$. Given any $\vx\in\R^n$, the distance between $\vx$ and $\calX$ is defined as
\[ \dist(\vx,\calX) := \inf_{\vz\in\calX}\|\vx - \vz\|. \]
Since $h$ can be nonsmooth, we utilize tools from generalized differentiation to formulate the optimality condition of~\eqref{eq:min h}. The (Fr\'{e}chet) subdifferential of $h$ at $\vx$ is defined as
\e
\partial h(\vx) : = \left\{\vd\in\R^n: \liminf_{\vy\rightarrow \vx}\frac{h(\vy) - h(\vx) - \langle \vd, \vy - \vx \rangle}{\|\vy - \vx\|}\geq 0   \right\},
\label{eq:subdifferential}\ee
where each $\vd\in\partial h(\vx)$ is called a subgradient of $h$ at $\vx$. We say that $\vx$ is a critical point of $h$ if $\vzero \in \partial h(\vx)$.

\subsection{Sharpness and Weak Convexity}
Since our goal is to consider a set of problems that can be solved by the SubGM with a linear rate of convergence, let us introduce two regularity notions for $h$ that are central to our study.
\begin{defi}[sharpness; cf.~\cite{burke1993weak}]
	We say that $h:\R^n \rightarrow \R$ is sharp with parameter $\alpha>0$ if
	\e
	h(\vx) - h^\star \geq \alpha \dist(\vx,\calX)
	\label{eq:sharp} 
	\ee
	for all $\vx \in \R^n$.
	\label{def:shaprness}
\end{defi}
\begin{defi}[weak convexity; see, e.g.,~\cite{V83}]
	We say that $h:\R^n \rightarrow \R$ is weakly convex with parameter $\tau\ge0$ if $\vx \mapsto h(\vx) + \tfrac{\tau}{2} \| \vx \|^2$ is convex.
	\label{def:weak convexity}
\end{defi}
It is worth noting that the function $h$ is weakly convex   with parameter $\tau \geq 0$ if and only if 
\e
h(\vw) - h(\vx) \geq   \langle \vd, \vw-\vx\rangle -    \frac{\tau}{2} \|\vw - \vx\|^2, \  \forall \ \vd \in \partial h(\vx).
\label{eq:weak convexity}
\ee
for any $\vw,\vx\in\R^n$, see, e.g., \cite[Lemma 2.1]{davis2019stochastic}. Indeed, this can be shown quickly by applying the convex subgradient inequality to $h(\vx) + \tfrac{\tau}{2} \| \vx \|^2$. 

Suppose that $h$ is sharp and weakly convex with parameters $\alpha > 0$ and $\tau \ge 0$, respectively. It is known that for any $\vx\notin \calX$ with $\dist(\vx,\calX) < \frac{2\alpha}{\tau}$, we have $\vzero\notin \partial h(\vx)$; i.e., $\vx$ is not a critical point of $h$~\cite[Lemma 3.1]{davis2018subgradient}. This suggests the possibility of finding a global minimum of $h$ by initializing local search algorithms with a point that is close to $\calX$. To explore such possibility, let us consider using the SubGM in \Cref{alg:sgd for general problem} to solve the nonsmooth nonconvex optimization problem \eqref{eq:min h}.

\begin{algorithm}[htb!]
	\caption{Subgradient Method (SubGM) for Solving~\eqref{eq:min h}}
	{\bf Initialization:} set $\vx_0$ and $\mu_0$;
	
	\begin{algorithmic}[1]
		\FOR{$k=0,1,\ldots$}
		\STATE compute a subgradient
		$\vd_k \in \partial h(\vx_k)$;
		\STATE update the step size $\mu_{k}$ according to a certain rule;
		\STATE update $\vx_{k+1} = \vx_{k} - \mu_{k}\vd_k$;
		\ENDFOR
	\end{algorithmic}
	\label{alg:sgd for general problem}
\end{algorithm}

\subsection{Convergence of SubGM for Sharp Weakly Convex Functions}

Unlike gradient descent, the SubGM with a constant step size may not converge to a critical point of a nonsmooth function in general, even when the function is convex~\cite{Shor85}. As a simple example, consider $h(x) = |x|$ and suppose that we take $x_0=0.01$ and $\mu_k = 0.02$ for all $k\ge0$ in~\Cref{alg:sgd for general problem}. Then, the iterates $\{x_k\}_{k\ge0}$ will oscillate between the two points $x_+=0.01$ and $x_-=-0.01$ and never converge to the global minimum $x^\star=0$.
 At best, one can only show that the SubGM with a constant step size will converge to a \emph{neighborhood} of the set of global optima of $h$ (with rate guarantees if $h$ satisfies additional regularity conditions); see, e.g.,~\cite{Shor85,NB01,B12,davis2018subgradient}. To ensure the convergence of the SubGM, a set of diminishing step sizes is generally needed~{\cite{Shor85,goffin1977convergence}}. As it turns out, for a sharp weakly convex function $h$, the SubGM with step sizes that are diminishing at a geometric rate can still be shown to converge linearly to a global minimum when initialized close to $\calX$. Specifically, let
\e
\kappa :=\sup\left\{\|\vd\| : \vd  \in \partial h(\vx),  \dist(\vx,\calX) < \frac{2\alpha}{\tau}\right\},
\label{eq:kappa}
\ee
which can be shown to satisfy $\kappa \ge \alpha$; cf.~\cite[Lemma 3.2]{davis2018subgradient}. Then, we have the following result:

\begin{thm}[local linear convergence of SubGM] 
Suppose that the function $h:\R^n \rightarrow \R$ is sharp and weakly convex with parameters $\alpha>0$ and $\tau\ge0$, respectively. Suppose further that the SubGM in~\Cref{alg:sgd for general problem} is initialized with a point $\vx_0$ satisfying $\dist(\vx_0,\calX)< \frac{2\alpha}{\tau}$ and uses the geometrically diminishing step sizes \label{thm:linear convergence of SubGM}
	\e
	\mu_k = \rho^k \mu_0,
	\label{eq:linear step size}
	\ee
	where the initial step size $\mu_0$ satisfies
	\e
	\mu_0 \leq \frac{\alpha^2}{2\tau \kappa^2}\left(1 - \left(\max\left\{\frac{\tau}{\alpha}\dist(\vx_0,\calX)-1,0\right\}\right)^2   \right)
	\label{eq:requirement on mu0}\ee
	and the decay rate $\rho$ satisfies
	\e
	1>	\rho \geq \underline\rho:= \sqrt{1-   \left( \frac{2\alpha}{\overline\dist_0} - \tau \right)\mu_0 + \frac{\kappa^2}{\overline\dist_0^2} \mu_0^2   }
	\label{eq:requirement on rho}	\ee
	with
	\e
	\overline\dist_0  = \max\left\{\dist(\vx_0,\calX), \mu_0\frac{\max\{\kappa^2,2\alpha^2\}}{\alpha} \right\}.
	\label{eq:bar dist0}\ee
	Then, the iterates $\{\vx_k\}_{k\ge0}$ generated by the SubGM will converge linearly to a point in $\calX$:
	\e
	\dist(\vx_k,\calX) \leq \rho^k \overline\dist_0, \ \forall k\geq 0.
	\label{eq:R linear decay}
	\ee
\end{thm}

 Before proceeding to the proof, we note that a similar result has been established in~\cite[Corollary 6.1]{davis2018subgradient}. Nevertheless, compared with \cite[Corollary 6.1]{davis2018subgradient}, which requires $\tfrac{\alpha}{\kappa} \le \sqrt{\tfrac{1}{2-\gamma}}$ and $\dist(\vx_0,\calX) \le \tfrac{\gamma\alpha}{\tau}$ for some $\gamma \in (0,1)$,~\Cref{thm:linear convergence of SubGM} is less restrictive and allows the larger initialization region $\dist(\vx_0,\calX)< \frac{2\alpha}{\tau}$. In particular, as $\tfrac{\alpha}{\kappa}$ tends to $1$, so does $\gamma$, and the decay rate $\rho$ in~\cite[Corollary 6.1]{davis2018subgradient} approaches $1$. Thus, one can no longer use~\cite[Corollary 6.1]{davis2018subgradient} to conclude that the SubGM converges linearly when $\tfrac{\alpha}{\kappa}=1$. By contrast, the linear convergence result in~\Cref{thm:linear convergence of SubGM} is still valid in this case.~\Cref{thm:linear convergence of SubGM} can be proven by refining the arguments in the proof of~\cite[Theorem 6.1]{davis2018subgradient}. 

\begin{proof}[Proof of \Cref{thm:linear convergence of SubGM}]
	We first show that $\underline \rho$ in \eqref{eq:requirement on rho} is well defined and satisfies $0< \underline \rho <1$. On one hand, we have
	\[ {\underline\rho}^2 > 1- \frac{2\alpha}{\overline\dist_0}\mu_0 \geq 1- \frac{2\alpha^2}{\max\{\kappa^2,2\alpha^2\}}\ge 0 \]
	by~\eqref{eq:bar dist0}. On the other hand, let $\theta(t) = \frac{\kappa^2}{t^2} \mu_0^2 - \left( \frac{2\alpha}{t} - \tau \right) \mu_0$ and note that $\underline\rho < 1$ if and only if $\theta(\overline\dist_0)<0$. Noting that $\alpha^2-\tau\mu_0\kappa^2>0$, the latter is equivalent to
	\e
	\frac{\alpha - \sqrt{\alpha^2 - \tau\mu_0\kappa^2}}{\tau} < \overline\dist_0 < \frac{\alpha+\sqrt{\alpha^2 - \tau\mu_0\kappa^2}}{\tau}.
	\label{eq:region for bar dist0}
	\ee
	To prove~\eqref{eq:region for bar dist0}, we first observe that
	\[ \frac{\alpha - \sqrt{\alpha^2 - \tau\mu_0\kappa^2}}{\tau} = \frac{\mu_0\kappa^2}{\alpha + \sqrt{\alpha^2 - \tau\mu_0\kappa^2}} < \frac{\mu_0\kappa^2}{\alpha} \le \mu_0\frac{\max\{\kappa^2,2\alpha^2\}}{\alpha} \le \overline\dist_0 \]
	by~\eqref{eq:bar dist0}. Now, due to \eqref{eq:requirement on mu0} and the fact that $\kappa\ge\alpha$, we have $\mu_0\frac{\max\{\kappa^2,2\alpha^2\}}{\alpha}\leq \frac{\alpha}{\tau}$. If $\dist(\vx_0,\calX)\leq \frac{\alpha}{\tau}$, then $\overline\dist_0 \le \frac{\alpha}{\tau} < \frac{\alpha + \sqrt{\alpha^2 - \tau\mu_0\kappa^2}}{\tau}$.  If $\dist(\vx_0,\calX)> \frac{\alpha}{\tau}$, then $\overline\dist_0 = \dist(\vx_0,\calX)$ and \eqref{eq:requirement on mu0} becomes
	\[ \mu_0 \leq \frac{-\tau\dist^2(\vx_0,\calX) + 2\alpha \dist(\vx_0,\calX)}{2\kappa^2} = \frac{-\tau{\overline\dist_0}^2 + 2\alpha {\overline\dist_0}}{2\kappa^2}. \]
	Noting that $\alpha^2-2\tau\mu_0\kappa^2\ge0$, we can solve the above quadratic inequality to get
	\[ \overline\dist_0 \le \frac{\alpha+\sqrt{\alpha^2 - 2\tau\mu_0\kappa^2}}{\tau} < \frac{\alpha+\sqrt{\alpha^2 - \tau\mu_0\kappa^2}}{\tau}. \]
	
	Now, we prove \eqref{eq:R linear decay} by induction. Since $\overline\dist_0\geq \dist(\vx_0,\calX)$ by~\eqref{eq:bar dist0}, it is clear that \eqref{eq:R linear decay} holds when $k=0$. Suppose then that \eqref{eq:R linear decay} holds at the $k$-th step. We compute
	\e\begin{split}
		\dist^2(\vx_{k+1},\calX) & \leq \| \vx_{k+1} - \calP_{\calX}(\vx_k) \|^2 =	\| \vx_{k} - \mu_k\vd_k -\calP_{\calX}(\vx_k) \|^2\\
		& = \dist^2(\vx_{k},\calX)  - 2 \mu_k\left\langle  \vx_{k}  -\calP_{\calX}(\vx_k),   \vd_k   \right\rangle + \mu_k^2 \| \vd_k \|^2 \\
		&\leq (1 + \tau\mu_k)\dist^2(\vx_{k},\calX) - 2 \mu_k (f(\vx_k) - f(\calP_{\calX}(\vx_k)))  + \mu_k^2  \| \vd_k \|^2\\
		& \leq (1 + \tau\rho^k\mu_0)\dist^2(\vx_{k},\calX) - 2 \rho^k\mu_0 \alpha \dist(\vx_k,\calX)  + \rho^{2k}\mu_0^2 \kappa^2,
	\end{split}
	\label{eq:dist k+1 to k}\ee
	where the second inequality utilizes \eqref{eq:weak convexity} and the last inequality is from \eqref{eq:sharp}, \eqref{eq:kappa}, and~\eqref{eq:linear step size}. Using~\eqref{eq:bar dist0} {and the fact that $\rho^k\in (0,1)$}, we have $\overline\dist_0 \geq \mu_0\frac{\max\{\kappa^2,2\alpha^2\}}{\alpha} \geq 2\mu_0\alpha \geq \frac{2\mu_0\alpha}{1+\tau \mu_0}$, which guarantees that the RHS in \eqref{eq:dist k+1 to k} attains its maximum at $\dist(\vx_{k},\calX) = \rho^k \overline\dist_0$ when $\dist(\vx_k,\calX) \leq \rho^k \overline\dist_0$. Thus, we have
	\begin{align*}
		\dist^2(\vx_{k+1},\calX)&\leq  \rho^{2k} \overline\dist_0^2 + \tau\rho^{2k}\mu_0 \overline\dist_0^2 -2\rho^{2k}\mu_0 \alpha  \overline\dist_0 + \rho^{2k}\mu_0^2 \kappa^2\\
		& =\rho^{2k}\overline\dist_0^2 \left(1 - \left(  \frac{2\alpha}{\overline\dist_0} -\tau \right)\mu_0 + \frac{\kappa^2}{\overline\dist_0^2} \mu_0^2  \right)\\
		& \leq  \rho^{2k}\overline\dist_0^2 \rho^2 =  \rho^{2(k+1)}\overline\dist_0^2,
	\end{align*}
	where the last inequality follows since $\rho \geq \underline\rho$. The proof is completed by induction.
\end{proof}

There are two factors in \eqref{eq:R linear decay}, namely, $\rho$ and $\overline \dist_0$, that determine the rate at which $\{\dist(\vx_k,\calX)\}_{k\ge0}$ tends to zero. Both factors depend crucially on the initial step size $\mu_0$. Indeed, when $\mu_0$ is large relative to $\dist(\vx_0,\calX)$, we have $\overline\dist_0  = \mu_0\frac{\max\{\kappa^2,2\alpha^2\}}{\alpha}$, which could be much larger than $\dist(\vx_0,\calX)$. Intuitively, although $\vx_0$ is close to $\calX$, $\vx_1$ could become far away from $\calX$ when $\mu_0$ is large. Thus, a larger $\overline\dist_0$ is needed in order for \eqref{eq:R linear decay} to hold. On the other hand, when $\mu_0$ is relatively small, we have $\overline\dist_0=\dist(\vx_0,\calX)$. To understand how $\mu_0$ affects $\underline\rho$, the best decay rate one can choose for $\rho$, let us consider the following cases:

\paragraph{Case I: $\dist(\vx_0,\calX)\leq \frac{\alpha}{\tau}$}
In this case, the initial step size $\mu_0$ satisfies $\mu_0 \leq \frac{\alpha^2}{2\tau \kappa^2}$; see~\eqref{eq:requirement on mu0}. If in addition we have $\mu_0 \leq \frac{\alpha\dist(\vx_0,\calX)}{\max\{\kappa^2,2\alpha^2\}}$, which implies that $\overline\dist_0=\dist(\vx_0,\calX)$, then $\underline\rho$ in \eqref{eq:requirement on rho} becomes
\e
\underline\rho =  \sqrt{1-   \left( \frac{2\alpha}{\dist(\vx_0,\calX)} - \tau \right)\mu_0 + \frac{\kappa^2}{\dist^2(\vx_0,\calX)} \mu_0^2}.
\label{eq:requirement on rho--dist0 large}\ee
{To see how $\underline \rho$ changes versus $\mu_0$, we consider the following two scenarios: (i) when $0\leq \dist(\vx_0,\calX) \leq \left( 1 - \frac{1}{\max\{1,2\alpha^2/\kappa^2\} }\right)\frac{2\alpha}{\tau}$, $\underline\rho$ increases as $\mu_0$ decreases and $\underline\rho \rightarrow 1$ as $\mu_0 \rightarrow 0$; (ii) when $\left( 1 - \frac{1}{\max\{1,2\alpha^2/\kappa^2\} }\right)\frac{2\alpha}{\tau} \leq \dist(\vx_0,\calX) \leq \frac{\alpha}{\tau}$, as $\mu_0$ decreases from $\frac{\alpha\dist(\vx_0,\calX)}{\max\{\kappa^2,2\alpha^2\}}$ to 0, $\underline\rho$ decreases and attains its minimum $\sqrt{1- \frac{(2\alpha - \tau\dist(\vx_0,\calX) )^2}{4\kappa^2}}$ at $\mu_0 = \frac{2\alpha \dist(\vx_0,\calX) - \tau \dist^2(\vx_0,\calX)}{2\kappa^2}$, after which $\underline\rho$ increases to 1 until $\mu_0$ reaches 0.}

On the other hand, if we have $\mu_0 \geq \frac{\alpha\dist(\vx_0,\calX)}{\max\{\kappa^2,2\alpha^2\}}$ and hence $\overline\dist_0  = \mu_0\frac{\max\{\kappa^2,2\alpha^2\}}{\alpha}$, then
\[
\underline\rho = \sqrt{1 -  \frac{2\alpha^2}{\max\{\kappa^2,2\alpha^2\}} + \frac{\kappa^2\alpha^2}{(\max\{\kappa^2,2\alpha^2\})^2} + \mu_0\tau },
\]
which increases as $\mu_0$ increases. We plot both $\overline\dist_0$ (red line) and $\underline \rho$ (blue line) as functions of $\mu_0$ in \Cref{fig:illustration_mu_rho_dist0}. Furthermore, the closer $\alpha$ and $\kappa$, the better the decay rate $\underline\rho$. In particular, when $\alpha=\kappa$ and $\mu_0 = \frac{\alpha\dist(\vx_0,\calX)}{\max\{\kappa^2,2\alpha^2\}}$, we get $\underline \rho = \sqrt{\frac{1}{4} + \frac{\tau\dist(\vx_0,\calX)}{2\alpha}}$, which approaches $\frac{1}{2}$ as $\dist(\vx_0,\calX)$ approaches zero.

\begin{figure}[!htp]
	\centering
	\includegraphics[width=5in]{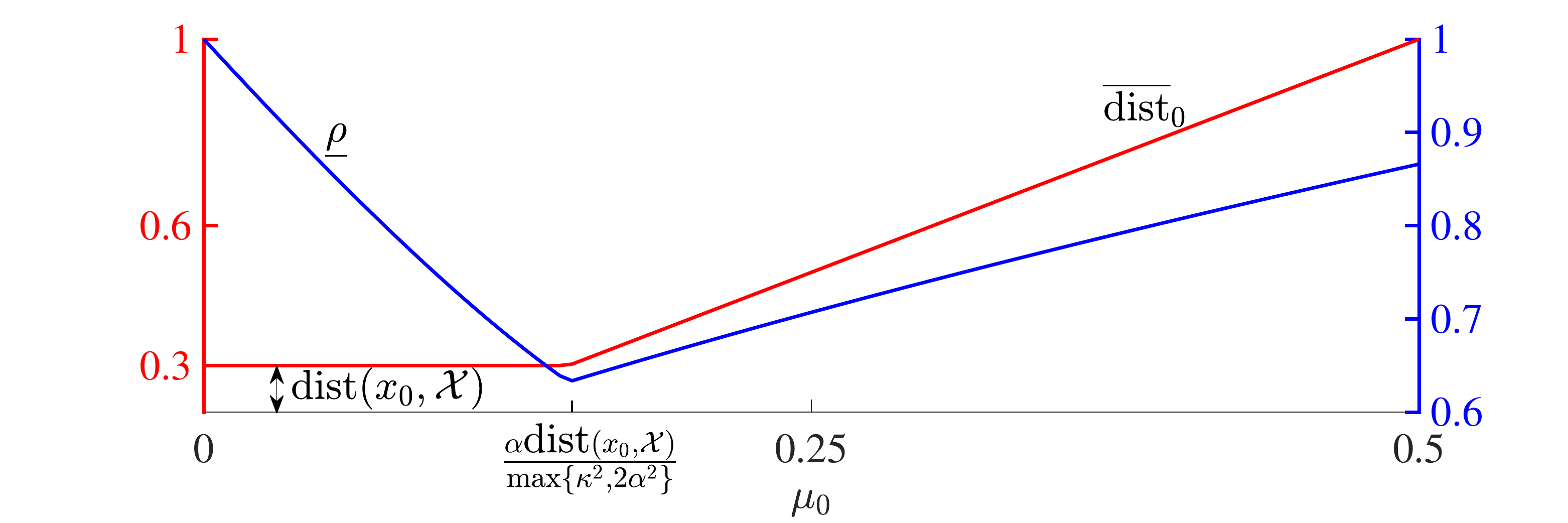}
	\caption{Illustration of $\underline \rho$ in \eqref{eq:requirement on rho} (blue line) and $\overline\dist_0$ in \eqref{eq:bar dist0} (red line) as a function of $\mu_0$ when $\dist(\vx_0,\calX)\leq \frac{\alpha}{\tau}$. For the purpose of illustration, we set $\alpha = \tau = \kappa =1$ and $\dist(\vx_0,\calX) = 0.3$.
	}\label{fig:illustration_mu_rho_dist0}
\end{figure}

In summary, the above discussion suggests that if $\dist(\vx_0,\calX)$ is known, then one can choose $\mu_0 = \frac{\alpha\dist(\vx_0,\calX)}{\max\{\kappa^2,2\alpha^2\}}$ so that both $\overline\dist_0$ and $\rho$ are made small; see \Cref{fig:illustration_mu_rho_dist0}. On the other hand, if $\dist(\vx_0,\calX)$ is not known \emph{a priori}, then one can choose a relatively large $\mu_0$ so that a small $\rho$ can be selected while at the same time the convergence of the SubGM can be ensured. In particular, if the parameters $\alpha,\tau,\kappa$ are known, then one can always choose
\[ \mu_0 = \frac{\alpha^2}{2\tau \kappa^2} \quad\mbox{and}\quad \underline \rho = \sqrt{1 -  \frac{2\alpha^2}{\max\{\kappa^2,2\alpha^2\}} + \frac{\kappa^2\alpha^2}{(\max\{\kappa^2,2\alpha^2\})^2} + \frac{\alpha^2}{2\kappa^2} }. \]

\paragraph{Case II: $\dist(\vx_0,\calX)> \frac{\alpha}{\tau}$}
In this case, \eqref{eq:requirement on mu0} implies that the initial step size $\mu_0$ satisfies $\mu_0 \leq \frac{-\tau\dist^2(\vx_0,\calX) + 2\alpha \dist(\vx_0,\calX)}{2\kappa^2}$, which decreases as $\dist(\vx_0,\calX)$ increases. Moreover, the best decay rate $\underline \rho$ takes the value in \eqref{eq:requirement on rho--dist0 large}, which again implies that a larger $\mu_0$ results in a smaller $\underline\rho$. Note that as $\dist(\vx_0,\calX)$ approaches $\frac{2\alpha}{\tau}$, the upper bound on $\mu_0$ goes to 0 and $\underline \rho$ goes to $1$, which means that the SubGM will converge very slowly.

{Before we proceed, it is worth elaborating on the implication of~\Cref{thm:linear convergence of SubGM} when $h$ is convex. In this case, we can take $\tau=0$, which, in view of~\eqref{eq:requirement on mu0}, shows that $\mu_0$ can be arbitrarily chosen. If we choose $\mu_0 \geq \tfrac{\alpha\dist(\vx_0,\calX)}{\max\{\kappa^2,2\alpha^2\}}$, then by~\eqref{eq:bar dist0} we have $\overline\dist_0  = \mu_0\frac{\max\{\kappa^2,2\alpha^2\}}{\alpha}$, which implies that the decay rate $\underline\rho$ satisfies
	\[
	\underline\rho = \sqrt{1 -  \frac{2\alpha^2}{\max\{\kappa^2,2\alpha^2\}} + \frac{\kappa^2\alpha^2}{(\max\{\kappa^2,2\alpha^2\})^2}} = \begin{cases}
	\sqrt{1 - \frac{\alpha^2}{\kappa^2}}, & \kappa^2 \ge 2\alpha^2, \\
	\frac{\kappa}{2\alpha}, 	& \kappa^2 < 2\alpha^2.
	\end{cases}
	\]
In particular, this is in line with the results in~\cite[Theorem 4.4]{goffin1977convergence}.}

\cmt{
\notexli{Keep the following proof and discussion in arXiv version but remove it in submitted version. }
\begin{proof}[Proof of \Cref{thm:linear convergence of SubGM}]
We first show that $\underline \rho$ in \eqref{eq:requirement on rho} is well defined and satisfies $0< \underline \rho <1$. On one hand, we have
\[ {\underline\rho}^2 > 1- \frac{2\alpha}{\overline\dist_0}\mu_0 \geq 1- \frac{2\alpha^2}{\max\{\kappa^2,2\alpha^2\}}\ge 0 \]
by~\eqref{eq:bar dist0}. On the other hand, let $\theta(t) = \frac{\kappa^2}{t^2} \mu_0^2 - \left( \frac{2\alpha}{t} - \tau \right) \mu_0$ and note that $\underline\rho < 1$ if and only if $\theta(\overline\dist_0)<0$. Noting that $\alpha^2-\tau\mu_0\kappa^2>0$, the latter is equivalent to
\e
\frac{\alpha - \sqrt{\alpha^2 - \tau\mu_0\kappa^2}}{\tau} < \overline\dist_0 < \frac{\alpha+\sqrt{\alpha^2 - \tau\mu_0\kappa^2}}{\tau}.
\label{eq:region for bar dist0}
\ee
To prove~\eqref{eq:region for bar dist0}, we first observe that
\[ \frac{\alpha - \sqrt{\alpha^2 - \tau\mu_0\kappa^2}}{\tau} = \frac{\mu_0\kappa^2}{\alpha + \sqrt{\alpha^2 - \tau\mu_0\kappa^2}} < \frac{\mu_0\kappa^2}{\alpha} \le \mu_0\frac{\max\{\kappa^2,2\alpha^2\}}{\alpha} \le \overline\dist_0 \]
by~\eqref{eq:bar dist0}. Now, due to \eqref{eq:requirement on mu0} and the fact that $\kappa\ge\alpha$, we have $\mu_0\frac{\max\{\kappa^2,2\alpha^2\}}{\alpha}\leq \frac{\alpha}{\tau}$. If $\dist(\vx_0,\calX)\leq \frac{\alpha}{\tau}$, then $\overline\dist_0 \le \frac{\alpha}{\tau} < \frac{\alpha + \sqrt{\alpha^2 - \tau\mu_0\kappa^2}}{\tau}$.  If $\dist(\vx_0,\calX)> \frac{\alpha}{\tau}$, then $\overline\dist_0 = \dist(\vx_0,\calX)$ and \eqref{eq:requirement on mu0} becomes
\[ \mu_0 \leq \frac{-\tau\dist^2(\vx_0,\calX) + 2\alpha \dist(\vx_0,\calX)}{2\kappa^2} = \frac{-\tau{\overline\dist_0}^2 + 2\alpha {\overline\dist_0}}{2\kappa^2}. \]
Noting that $\alpha^2-2\tau\mu_0\kappa^2\ge0$, we can solve the above quadratic inequality to get
\[ \overline\dist_0 \le \frac{\alpha+\sqrt{\alpha^2 - 2\tau\mu_0\kappa^2}}{\tau} < \frac{\alpha+\sqrt{\alpha^2 - \tau\mu_0\kappa^2}}{\tau}. \]

Now, we prove \eqref{eq:R linear decay} by induction. Since $\overline\dist_0\geq \dist(\vx_0,\calX)$ by~\eqref{eq:bar dist0}, it is clear that \eqref{eq:R linear decay} holds when $k=0$. Suppose then that \eqref{eq:R linear decay} holds at the $k$-th step. We compute
\e\begin{split}
\dist^2(\vx_{k+1},\calX) & \leq \| \vx_{k+1} - \calP_{\calX}(\vx_k) \|^2 =	\| \vx_{k} - \mu_k\vd_k -\calP_{\calX}(\vx_k) \|^2\\
& = \dist^2(\vx_{k},\calX)  - 2 \mu_k\left\langle  \vx_{k}  -\calP_{\calX}(\vx_k),   \vd_k   \right\rangle + \mu_k^2 \| \vd_k \|^2 \\
& \leq (1 + \tau\rho^k\mu_0)\dist^2(\vx_{k},\calX) - 2 \rho^k\mu_0 \alpha \dist(\vx_k,\calX)  + \rho^{2k}\mu_0^2 \kappa^2,
\end{split}
\label{eq:dist k+1 to k}\ee
where the second inequality utilizes \eqref{eq:consequence of sharp and weak convex}, \eqref{eq:kappa}, and~\eqref{eq:linear step size}. According to \eqref{eq:bar dist0} {and $\rho^k\in (0,1)$}, we have $\overline\dist_0 \geq \mu_0\frac{\max\{\kappa^2,2\alpha^2\}}{\alpha} \geq 2\mu_0\alpha \geq \frac{2\mu_0\alpha}{1+\tau \mu_0}$, which guarantees that the RHS in \eqref{eq:dist k+1 to k} attains its maximum at $\dist(\vx_{k},\calX) = \rho^k \overline\dist_0$ when $\dist(\vx_k,\calX) \leq \rho^k \overline\dist_0$. Thus, we have
\begin{align*}
\dist^2(\vx_{k+1},\calX)&\leq  \rho^{2k} \overline\dist_0^2 + \tau\rho^{2k}\mu_0 \overline\dist_0^2 -2\rho^{2k}\mu_0 \alpha  \overline\dist_0 + \rho^{2k}\mu_0^2 \kappa^2\\
& =\rho^{2k}\overline\dist_0^2 \left(1 - \left(  \frac{2\alpha}{\overline\dist_0} -\tau \right)\mu_0 + \frac{\kappa^2}{\overline\dist_0^2} \mu_0^2  \right)\\
& \leq  \rho^{2k}\overline\dist_0^2 \rho^2 =  \rho^{2(k+1)}\overline\dist_0^2,
\end{align*}
where the last inequality follows since $\rho \geq \underline\rho$. The proof is completed by induction.
\end{proof}	

There are two factors in \eqref{eq:R linear decay}, namely, $\rho$ and $\overline \dist_0$, that determine the rate at which $\{\dist(\vx_k,\calX)\}_{k\ge0}$ tends to zero. Both factors depend crucially on the initial step size $\mu_0$. Indeed, when $\mu_0$ is large relative to $\dist(\vx_0,\calX)$, we have $\overline\dist_0  = \mu_0\frac{\max\{\kappa^2,2\alpha^2\}}{\alpha}$, which could be much larger than $\dist(\vx_0,\calX)$. Intuitively, although $\vx_0$ is close to $\calX$, $\vx_1$ could become far away from $\calX$ when $\mu_0$ is large. Thus, a larger $\overline\dist_0$ is needed in order for \eqref{eq:R linear decay} to hold. On the other hand, when $\mu_0$ is relatively small, we have $\overline\dist_0=\dist(\vx_0,\calX)$. To understand how $\mu_0$ affects $\underline\rho$, the best decay rate one can choose for $\rho$, let us consider the following cases:

\paragraph{Case I: $\dist(\vx_0,\calX)\leq \frac{\alpha}{\tau}$}
In this case, the initial step size $\mu_0$ satisfies $\mu_0 \leq \frac{\alpha^2}{2\tau \kappa^2}$; see~\eqref{eq:requirement on mu0}. If in addition we have $\mu_0 \leq \frac{\alpha\dist(\vx_0,\calX)}{\max\{\kappa^2,2\alpha^2\}}$, which implies that $\overline\dist_0=\dist(\vx_0,\calX)$, then $\underline\rho$ in \eqref{eq:requirement on rho} becomes
\e
\underline\rho =  \sqrt{1-   \left( \frac{2\alpha}{\dist(\vx_0,\calX)} - \tau \right)\mu_0 + \frac{\kappa^2}{\dist^2(\vx_0,\calX)} \mu_0^2}.
\label{eq:requirement on rho--dist0 large}\ee
 {To see how $\underline \rho$ changes against $\mu_0$, we consider the following two scenarios: $(i)$ when $0\leq \dist(\vx_0,\calX) \leq \left( 1 - \frac{1}{\max\{1,2\alpha^2/\kappa^2\} }\right)\frac{2\alpha}{\tau}$, we have that $\underline\rho$ is monotonically increasing  as $\mu_0$ decreases and $\underline\rho \rightarrow 1$ when  $\mu_0 \rightarrow 0$; $(ii)$ when  $\left( 1 - \frac{1}{\max\{1,2\alpha^2/\kappa^2\} }\right)\frac{2\alpha}{\tau} \leq \dist(\vx_0,\calX) \leq \frac{\alpha}{\tau}$, in the process of decreasing $\mu_0$ from $\frac{\alpha\dist(\vx_0,\calX)}{\max\{\kappa^2,2\alpha^2\}}$ to  0,  $\underline\rho$ will decrease in the first stage and attain its minimum $\sqrt{1- \frac{(2\alpha - \tau\dist(\vx_0,\calX) )^2}{4\kappa^2}}$ at $\mu_0 = \frac{2\alpha \dist(\vx_0,\calX) - \tau \dist^2(\vx_0,\calX)}{2\kappa^2}$, after which $\underline\rho$ will monotonically increase to 1 until $\mu_0$ decreases to 0.   }

 On the other hand, if we have $\mu_0 \geq \frac{\alpha\dist(\vx_0,\calX)}{\max\{\kappa^2,2\alpha^2\}}$ and hence $\overline\dist_0  = \mu_0\frac{\max\{\kappa^2,2\alpha^2\}}{\alpha}$, then
\[
\underline\rho = \sqrt{1 -  \frac{2\alpha^2}{\max\{\kappa^2,2\alpha^2\}} + \frac{\kappa^2\alpha^2}{(\max\{\kappa^2,2\alpha^2\})^2} + \mu_0\tau },
\]
which increases as $\mu_0$ increases. We plot both $\overline\dist_0$ (red line) and $\underline \rho$ (blue line) as functions of $\mu_0$ in \Cref{fig:illustration_mu_rho_dist0}. Furthermore, the closer $\alpha$ and $\kappa$, the better the decay rate $\underline\rho$. In particular, when $\alpha=\kappa$ and $\mu_0 = \frac{\alpha\dist(\vx_0,\calX)}{\max\{\kappa^2,2\alpha^2\}}$, we get $\underline \rho = \sqrt{\frac{1}{4} + \frac{\tau\dist(\vx_0,\calX)}{2\alpha}}$, which approaches $\frac{1}{2}$ as $\dist(\vx_0,\calX)$ approaches zero.

\begin{figure}[!htp]
\centering
\includegraphics[width=5in]{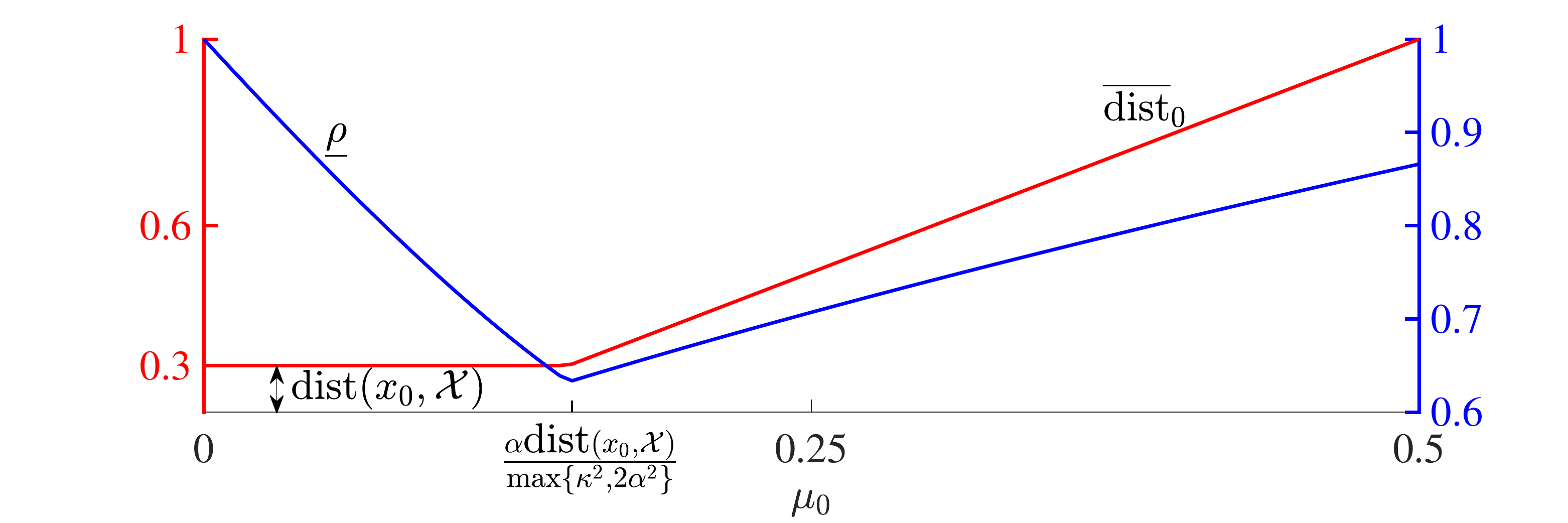}
\caption{Illustration of $\underline \rho$ in \eqref{eq:requirement on rho} (blue line) and $\overline\dist_0$ in \eqref{eq:bar dist0} (red line) as a function of $\mu_0$ when $\dist(\vx_0,\calX)\leq \frac{\alpha}{\tau}$. For the purpose of illustration, we set $\alpha = \tau = \kappa =1$ and $\dist(\vx_0,\calX) = 0.3$.
}\label{fig:illustration_mu_rho_dist0}
\end{figure}

In summary, the above discussion suggests that if $\dist(\vx_0,\calX)$ is known, then one can choose $\mu_0 = \frac{\alpha\dist(\vx_0,\calX)}{\max\{\kappa^2,2\alpha^2\}}$ so that both $\overline\dist_0$ and $\rho$ are made small; see \Cref{fig:illustration_mu_rho_dist0}. On the other hand, if $\dist(\vx_0,\calX)$ is not known \emph{a priori}, then one can choose a relatively large $\mu_0$ so that a small $\rho$ can be selected while at the same time the convergence of the SubGM can be ensured. In particular, if the parameters $\alpha,\tau,\kappa$ are known, then one can always choose
\[ \mu_0 = \frac{\alpha^2}{2\tau \kappa^2} \quad\mbox{and}\quad \underline \rho = \sqrt{1 -  \frac{2\alpha^2}{\max\{\kappa^2,2\alpha^2\}} + \frac{\kappa^2\alpha^2}{(\max\{\kappa^2,2\alpha^2\})^2} + \frac{\alpha^2}{2\kappa^2} }. \]

\paragraph{Case II: $\dist(\vx_0,\calX)> \frac{\alpha}{\tau}$}
In this case, \eqref{eq:requirement on mu0} implies that the initial step size $\mu_0$ satisfies $\mu_0 \leq \frac{-\tau\dist^2(\vx_0,\calX) + 2\alpha \dist(\vx_0,\calX)}{2\kappa^2}$, which decreases as $\dist(\vx_0,\calX)$ increases. Moreover, the best decay rate $\underline \rho$ takes the value in \eqref{eq:requirement on rho--dist0 large}, which again implies that a larger $\mu_0$ results in a smaller $\underline\rho$. Note that as $\dist(\vx_0,\calX)$ approaches $\frac{2\alpha}{\tau}$, the upper bound on $\mu_0$ goes to 0 and $\underline \rho$ goes to $1$, which means that the SubGM will converge very slowly.

\notexli{I do not know how to answer Q1 of R2... It looks better to keep the current statement and add one sentence somewhere to say that our assumption in this paper is basically to let $\calA$ has $\ell_1/\ell_2$-RIP and the fraction of outliers is slightly less than $\frac{1}{2}$ which is related to the RIP constant. Otherwise, it seems we need to revise hugely.}
}

\section{Nonconvex Robust Low-Rank Matrix Recovery: Symmetric Positive Semidefinite (PSD) Case}
\label{sec:geometry analysis}
In the last section we saw that the SubGM with suitable initialization and step sizes converges linearly to a global minimum of a sharp weakly convex function. Naturally, it is of interest to identify concrete problems that possess these two regularity properties. In this section we focus on the robust low-rank matrix recovery problem~\eqref{eq:rms factorization} and establish, for the first time, a connection between the exact recovery condition of $\ell_1/\ell_2$-RIP and the regularity properties of sharpness and weak convexity of the objective function $f$ in~\eqref{eq:rms factorization}. Specifically, we first show that if the fraction of outliers is slightly less than $\frac{1}{2}$ and certain measurement operators arising from the measurement model~\eqref{eq:rms model} possess the $\ell_1/\ell_2$-RIP, then the sharpness condition in \Cref{def:shaprness} holds for~\eqref{eq:rms factorization}. Consequently, all global minima of~\eqref{eq:rms factorization} lead to the exact recovery of the ground-truth matrix $\mX^\star$. We then show that~\eqref{eq:rms factorization} also satisfies the weak convexity condition in \Cref{def:weak convexity}.  Hence, by the convergence result (\Cref{thm:linear convergence of SubGM}) in the last section, we conclude that the SubGM can be utilized to find a global minimum of~\eqref{eq:rms factorization} efficiently.

To begin, let us collect some preparatory results. Let $\mX^\star = \mU^\star\mU^{\star\T}$ be a factorization of $\mX^\star$, where $\mU^\star\in\R^{n\times r}$. Note that for any $\mR\in\calO_r$, we have $\mX^\star = \mU^\star\mR (\mU^\star\mR)^\T$. Thus, all elements in the set
\[
\calU:=\left\{\mU^\star\mR: \mR\in\calO_r   \right\}
\]
are valid factors of $\mX^\star$. Furthermore, it is clear that the function $f$ in \eqref{eq:rms factorization} is constant on the set $\calU$. The following result connects $\dist(\mU,\calU)$ and the distance between $\mU\mU^\T$ and $\mU^\star \mU^{\star \T}$ for any given $\mU \in \R^{n\times r}$:
\begin{lem}[{\cite[Lemma 5.4]{tu2015low}}]
	Given any $\mU^\star \in \R^{n\times r}$, define $\mX^\star = \mU^\star \mU^{\star \T}$. Then, for any $\mU\in\R^{n\times r}$, we have
	\[
	2\left( \sqrt{2}-1 \right)\sigma_{r}^2(\mX^\star) \dist^2(\mU,\calU) \leq \|\mU\mU^\T - \mU^\star \mU^{\star \T}\|_F^2,
	\]
	where  $\sigma_r$ denotes the $r$-th largest singular value.
	\label{lem:Procrustes problem}
\end{lem}

\subsection{$\ell_1/\ell_2$-Restricted Isometry Property}
\label{subsec:L1_2 RIP}
Since the $\ell_1/\ell_2$-RIP~\cite{ZHZ13,chen2015exact,YS16} of the linear measurement operator $\calA:\R^{n\times n} \rightarrow \R^m$ in~\eqref{eq:rms factorization} plays an important role in our subsequent analysis, let us first provide a condition under which $\calA$ will possess such property.   Recall that $\calA$ can be specified by a collection of $m$ $n\times n$ matrices $\mA_1,\ldots,\mA_m$. In other words, given any $\mX \in \R^{n\times n}$, we have $\mathcal{A}(\mX) = \left( \langle \mA_1, \mX \rangle, \ldots, \langle  \mA_m, \mX \rangle \right)$.   We now show that if $\mA_1,\ldots,\mA_m$ have independent and identically distributed (\emph{i.i.d.}) standard Gaussian entries, then $\calA$ will possess the $\ell_1/\ell_2$-RIP with high probability.

\begin{prop}[$\ell_1/\ell_2$-RIP of Gaussian measurement operators]
	Let $r\ge1$ be given. Suppose that $m\gtrsim nr$ and the matrices $\mA_1,\ldots,\mA_m \in \R^{n\times n}$ defining the linear measurement operator $\calA$ have \emph{i.i.d.} standard Gaussian entries. Then, for any $0<\delta<\sqrt{\frac{2}{\pi}}$, there exists a universal constant $c>0$ such  that with  probability exceeding $1-\exp(-c\delta^2 m)$, $\calA$ will possess the $\ell_1/\ell_2$-RIP; i.e., the inequalities
	\e \label{eq:L1-2 RIP for Gaussian}
	\left( \sqrt{\frac{2}{\pi}} - \delta \right) \|\mX\|_F \leq	\frac{1}{m}  \|\calA( \mX)\|_1  \leq  \left( \sqrt{\frac{2} {\pi}} + \delta \right) \|\mX\|_F
	\ee
	hold for any rank-$2r$ matrix $\mX \in \R^{n\times n}$.
	\label{thm: L1-2 RIP for Gaussian}
\end{prop}

The proof of \Cref{thm: L1-2 RIP for Gaussian} is given in \Cref{sec:prf l1l2 rip}. It is worth noting that similar $\ell_1/\ell_2$-RIPs hold for other types of measurement operators such as the quadratic measurement operators in~\cite{chen2015exact} and those defined by sub-Gaussian matrices. Thus, although our results are stated for Gaussian measurement operators, they can be readily extended to cover other measurement operators that possess similar RIPs.

\subsection{Sharpness and Exact Recovery} \label{subsec:sharp}
Assuming that the linear measurement operator $\calA$ possesses the $\ell_1/\ell_2$-RIP~\eqref{eq:L1-2 RIP for Gaussian}, our first goal is to identify further conditions on the measurement model~\eqref{eq:rms model} so that any global minimum $\mU^\star$ of~\eqref{eq:rms factorization} can be used to recover the ground-truth matrix $\mX^\star$ via $\mU^\star\mU^{\star \T} = \mX^\star$.  Towards that end, let $\Omega\subseteq\{1,\ldots,m\}$ denote the support of the outlier vector $\vs^\star$ and $\Omega^c = \{1,\ldots,m\} \setminus \Omega$. Furthermore, let $p = \frac{|\Omega|}{m}$ be the fraction of outliers in $\vy$. Throughout, we do not make any assumption on the location of the non-zero entries of $\vs^\star$. Instead, we assume that $\calA_{\Omega^c}$, the linear operator defined by the matrices in $\{\mA_i:i \in \Omega^c\}$, also possesses the $\ell_1/\ell_2$-RIP; i.e., we have
\e
\left( \sqrt{\frac{2}{\pi}} - \delta \right)  \|\mX\|_F \leq \frac{1}{m(1- p)} \left\| \left[\calA\left(\mX\right)\right]_{\Omega^c} \right\|_1 \leq \left( \sqrt{\frac{2}{\pi}} + \delta \right) \|\mX\|_F
\label{eq:RIP partial A}\ee
for any rank-$2r$ matrix $\mX$. When each $\mA_i$ is generated with \emph{i.i.d.} standard Gaussian entries, \Cref{thm: L1-2 RIP for Gaussian} implies that $\calA_{\Omega^c}$ will satisfy \eqref{eq:RIP partial A} with high probability as long as $p$ is a constant. This follows from the fact that $|\Omega^c| = (1-p)m \gtrsim nr$ if $m\gtrsim nr$. 

\begin{prop}[sharpness and exact recovery with outliers: PSD case]
Let $0<\delta<\frac{1}{3}\sqrt{\frac{2}{\pi}}$ be given. Suppose that the fraction of outliers $p$ satisfies
\e
p< \frac{1}{2}- \frac{\delta }{\sqrt{{2}/{\pi}} - \delta},
\label{eq:p and delta}
\ee	
and that the linear operators $\calA$ and $\calA_{\Omega^c}$ possess the $\ell_1/\ell_2$-RIP~\eqref{eq:L1-2 RIP for Gaussian} and~\eqref{eq:RIP partial A}, respectively. Then, the objective function $f$ in \eqref{eq:rms factorization} satisfies
\[ 
f(\mU) - f(\mU^\star) \geq \alpha  \dist(\mU,\calU)
\] 
for any $\mU \in \R^{n\times r}$, where
	\e\label{eq:sharness parameter}
	\alpha = \sqrt{2\left(\sqrt{2}-1\right)} \left(2(1-p) \left( \sqrt{\frac{2}{\pi}} - \delta \right) - \left( \sqrt{\frac{2}{\pi}} + \delta \right) \right) \sigma_r(\mX^\star) >0.
	\ee
In particular, the set $\calU$ is precisely the set of global minima of \eqref{eq:rms factorization} and the objective function $f$ is sharp with parameter $\alpha>0$.
	\label{lem:sharpness rms}
\end{prop}

\begin{proof}[Proof of \Cref{lem:sharpness rms}]
	Using~\eqref{eq:rms model} and~\eqref{eq:rms factorization}, we compute
	\e\begin{split}
		&f(\mU) - f(\mU^\star) =	\frac{1}{m} \left\| \calA\left(\mU\mU^\T - \mX^\star\right) -\vs^\star \right\|_1 - \frac{1}{m} \left\| \vs^\star \right\|_1 \\
		&= \frac{1}{m} \left\| \left[\calA\left(\mU\mU^\T - \mX^\star\right)\right]_{\Omega^c} \right\|_1  +	\frac{1}{m} \left\| \left[\calA\left(\mU\mU^\T - \mX^\star\right)\right]_{\Omega} -\vs^\star \right\|_1 - \frac{1}{m} \left\| \vs^\star \right\|_1  \\
		&\geq\frac{1}{m} \left\| \left[\calA\left(\mU\mU^\T - \mX^\star\right)\right]_{\Omega^c} \right\|_1 -  \frac{1}{m} \left\| \left[\calA\left(\mU\mU^\T - \mX^\star\right)\right]_{\Omega} \right\|_1\\
		& =\frac{2}{m} \left\| \left[\calA\left(\mU\mU^\T - \mX^\star\right)\right]_{\Omega^c} \right\|_1 - \frac{1}{m} \left\|\calA\left(\mU\mU^\T - \mU^\star\mU^{\star \T} \right) \right\|_1\\
		&\geq \left(2(1-p) \left( \sqrt{\frac{2}{\pi}} - \delta \right) - \left( \sqrt{\frac{2}{\pi}} + \delta \right) \right) \left\| \mU^\star\mU^\starT  - \mU\mU^\T \right\|_F \\
		& \geq \alpha \dist(\mU,\calU),
	\end{split}	
	\nonumber		\ee
	where the second inequality follows from the $\ell_1/\ell_2$-RIP of $\calA$ and $\calA_{\Omega^c}$ and the last inequality follows from \Cref{lem:Procrustes problem}. The characterization of the set of global minima of~\eqref{eq:rms factorization} follows immediately from the above inequality and the choice of $p$ in \eqref{eq:p and delta}.
\end{proof}

One interesting consequence of~\Cref{lem:sharpness rms} is that for the robust low-rank matrix recovery problem~\eqref{eq:rms factorization}, the sharpness condition (which characterizes the geometry of the optimization problem around the set of global minima) coincides with the exact recovery property (which is of statistical nature).  
Moreover, {condition~\eqref{eq:p and delta} suggests that the smaller $\delta$ is, the higher the outlier ratio $p$ can be.} On the other hand, given an outlier ratio $p$, condition \eqref{eq:p and delta} requires that $\delta < \sqrt{\frac{2}{\pi}} - \frac{ \sqrt{2/\pi} }{3/2 - p}$, which indirectly imposes a condition on the number of measurements $m$. {Indeed,~\Cref{thm: L1-2 RIP for Gaussian} implies that in order for a Gaussian measurement operator $\calA$ to possess the $\ell_1/\ell_2$-RIP with positive probability, we need $m\gtrsim nr \Big/ \left( \sqrt{\tfrac{2}{\pi}} - \tfrac{ \sqrt{2/\pi} }{3/2 - p} \right)^2$ measurements. Putting it another way, the larger the number of measurements $m$ is, the higher the outlier ratio $p$ can be.} We shall elaborate on this point with experiments in \Cref{sec:experiments}.


\subsection{Weak Convexity} \label{subsec:weak-cvx}
In the last subsection we established the sharpness of \eqref{eq:rms factorization} and showed that any of its global minimum will lead to the exact recovery of the ground-truth matrix $\mX^\star$, even when the fraction of outliers is up to almost $\frac{1}{2}$.  In this subsection we further establish the weak convexity of~\eqref{eq:rms factorization}, thus opening up the possibility of using the machinery developed in \Cref{sec:preliminaries} to obtain provable convergence guarantees for the SubGM when it is applied to solve~\eqref{eq:rms factorization}. Towards that end, we note that the $\ell_1$-norm, being a convex function, is subdifferentially regular~\cite[Example 7.27]{RW04} (see~\cite[Definition 7.25]{RW04} for the definition of subdifferential regularity). Hence, by the chain rule for subdifferentials of subdifferentially regular functions~\cite[Corollary 8.11 and Theorem 10.6]{RW04}, we have
\e
\partial f(\mU) = \frac{1}{m}\left[ \left( \calA^*\left( \Sign \left( \calA(\mU\mU^\T) - \vy \right) \right) \right)^\T\mU +   \calA^* \left( \Sign\left( \calA(\mU\mU^\T) - \vy \right) \right) \mU\right].
\label{eq:subdifferential f}\ee
We are now ready to prove the following result. Note that the weak convexity parameter $\tau$ in~\eqref{eq:weak convexity parameter} is independent of the fraction of outliers.
\begin{prop}[weak convexity: PSD case]\label{prop:no outlier weak convex}
	Suppose that the measurement operator $\calA$ satisfies the $\ell_1/\ell_2$-RIP \eqref{eq:L1-2 RIP for Gaussian}. Then, the objective function $f$ in \eqref{eq:rms factorization} is weakly convex with parameter
	\e
	\tau = 2 \left( \sqrt{\frac{2}{\pi}} + \delta \right).
	\label{eq:weak convexity parameter}\ee
\end{prop}
\begin{proof}[Proof of \Cref{prop:no outlier weak convex}]
	For any $\mU', \mU\in \R^{n\times r}$, let $\mDelta = \mU' - \mU$. Then, we have
	\begin{align*}
	&f(\mU') = \frac{1}{m} \left\|\calA(\mU'\mU'^\T - \mX^\star) - \vs^\star\right\|_1\\
	&= \frac{1}{m} \left\|\calA( \mU\mU^\T - \mX^\star + \mU\mDelta^\T + \mDelta\mU^\T + \mDelta\mDelta^\T)  - \vs^\star\right\|_1\\
	&\geq  \frac{1}{m} \left\| \calA(\mU\mU^\T - \mX^\star + \mU\mDelta^\T + \mDelta\mU^\T ) - \vs^\star \right\|_1  - \frac{1}{m} \left\|\calA( \mDelta\mDelta^\T )\right\|_1\\
	&\geq  \frac{1}{m} \left\| \calA(\mU\mU^\T - \mX^\star + \mU\mDelta^\T + \mDelta\mU^\T ) - \vs^\star \right\|_1 - \left( \sqrt{\frac{2}{\pi}} + \delta \right) \left\|\mDelta\mDelta^\T \right\|_F\\
	&\geq  f(\mU) +  \frac{1}{m} \left\langle \vd, \calA( \mU\mDelta^\T + \mDelta\mU^\T ) \right\rangle   - \frac{\tau}{2} \| \mDelta \|_F^2
	\end{align*}
	{for any $\vd \in \Sign(\calA(\mU\mU^\T) - \vy)$},	where the second inequality follows from the $\ell_1/\ell_2$-RIP of $\calA$ and the last inequality is due to the convexity of the $\ell_1$-norm and $\|\mDelta\mDelta^\T\|_F \leq \|\mDelta\|_F^2$. Substituting~\eqref{eq:subdifferential f} into the above equation gives
\[ f(\mU') \geq  f(\mU)  + \langle \mD, \mU'-\mU\rangle - \frac{\tau}{2} \| \mU' - \mU \|_F^2, \ \forall \ \mD\in \partial f(\mU). \]
This completes the proof.
\end{proof}

\subsection{Putting Everything Together\label{subsec:summarize symmetric guarantees}}

With the results in~\Cref{subsec:sharp} and \Cref{subsec:weak-cvx} in place, in order to show that the SubGM enjoys the convergence guarantees in~\Cref{thm:linear convergence of SubGM} when applied to the robust low-rank matrix recovery problem~\eqref{eq:rms factorization}, it remains to determine $\kappa$, the bound on the norm of any subgradient of $f$ in a neighborhood of $\calU$; see~\eqref{eq:kappa}. This is established by the following result:
\begin{prop}[bound on subgradient norm: PSD case] 
{Suppose that the measurement operator $\calA$ satisfies the $\ell_1/\ell_2$-RIP \eqref{eq:L1-2 RIP for Gaussian}. Then, for any $\mU\in \R^{n\times r}$ satisfying $\dist(\mU,\calU) \leq \frac{2\alpha}{\tau}$, we have} \label{lem:symmetric bound kappa} 
\e
{\|\mD\|_F \leq \kappa = 2 \left( \sqrt{\frac{2}{\pi}}+ \delta \right) \left( \|\mU^\star\|_F + \frac{2\alpha}{\tau} \right),  \ \forall \mD\in \partial f(\mU).}
\label{eq:uniform bound for subgradient}
\ee
\end{prop}
\begin{proof}[Proof of \Cref{lem:symmetric bound kappa}]
Recall from~\eqref{eq:subdifferential} that
\e
\liminf_{\mU'\rightarrow \mU}\frac{f(\mU') - f(\mU) - \langle \mD, \mU' - \mU \rangle}{\|\mU' - \mU\|_F}\geq 0
\label{eq:subgradient for f}\ee
for any $\mD \in \partial f(\mU)$. Now, for any $\mU'\in\R^{n\times r}$,
\begin{align*}
	\left|f(\mU') - f(\mU)\right| &= \frac{1}{m}\left| \left\| \vy - \calA(\mU'\mU'^\T) \right\|_1 - \left\| \vy - \calA(\mU\mU^\T) \right\|_1 \right| \\
	& \leq \frac{1}{m} \left\| \calA(\mU'\mU'^\T - \mU\mU^\T) \right\|_1\\
	& \leq \left( \sqrt{\frac{2}{\pi}}+ \delta \right) \left\| \mU'\mU'^\T - \mU\mU^\T \right\|_F\\
	& = \left( \sqrt{\frac{2}{\pi}}+ \delta \right) \left\| (\mU' - \mU)  \mU^\T + \mU' (\mU' - \mU)^\T \right\|_F \\
	& \leq \left( \sqrt{\frac{2}{\pi}}+ \delta \right) (\|\mU\| + \|\mU'\|) \|\mU' - \mU\|_F,
\end{align*}
where the second inequality follows from the $\ell_1/\ell_2$-RIP of $\calA$. It follows that
\begin{align*}
\liminf_{\mU'\rightarrow \mU}\frac{\left| f(\mU')-f(\mU) \right| }{\|\mU - \mU'\|_F} &\leq \lim_{\mU'\rightarrow \mU}\frac{(\sqrt{{2}/{\pi}}+ \delta)(\|\mU\| + \|\mU'\|)\|\mU' - \mU\|_F}{\|\mU' - \mU\|_F} \\
&= 2 \left( \sqrt{\frac{2}{\pi}}+ \delta \right) \|\mU\|.
\end{align*}
Upon taking $\mU' = \mU + t \mD$, $t\rightarrow 0$ and invoking \eqref{eq:subgradient for f}, we get
\[
\|\mD\|_F \leq 2\left( \sqrt{\frac{2}{\pi}}+ \delta \right) \|\mU\|, \ \forall \ \mD\in \partial f(\mU).
\]
To complete the proof, it remains to note that for any $\mU \in \R^{n\times r}$ satisfying $\dist(\mU,\calU)\leq \frac{2\alpha}{\tau}$, where $\alpha,\tau$ are given in~\eqref{eq:sharness parameter},~\eqref{eq:weak convexity parameter}, respectively, the triangle inequality yields $ \|\mU\|\leq  \|\mU^\star\|_F + \frac{2\alpha}{\tau}$.
\end{proof}

By collecting \Cref{lem:sharpness rms}, \Cref{prop:no outlier weak convex}, and \Cref{lem:symmetric bound kappa} together and invoking \Cref{thm:linear convergence of SubGM}, we obtain the following guarantees for the SubGM\footnote{In practice, we can just take $\Sign(0)=0$ when applying the SubGM to solve~\eqref{eq:rms factorization}.} when it is applied to the robust low-rank matrix recovery problem \eqref{eq:rms factorization}:

\begin{thm}[nonconvex robust low-rank matrix recovery: PSD case] \label{thm:rms}
Consider the measurement model~\eqref{eq:rms model}, where $\mX^\star$ is an $n\times n$ rank-$r$ symmetric positive semidefinite matrix. Let $0<\delta<\frac{1}{3}\sqrt{\frac{2}{\pi}}$ be given. Suppose that the fraction of outliers $p$ in the measurement vector $\vy$ satisfies~\eqref{eq:p and delta}, and that the linear operators $\calA$, $\calA_{\Omega^c}$ possess the $\ell_1/\ell_2$-RIP~\eqref{eq:L1-2 RIP for Gaussian},~\eqref{eq:RIP partial A}, respectively. Let $\alpha$, $\tau$, and $\kappa$ be given by~\eqref{eq:sharness parameter}, \eqref{eq:weak convexity parameter}, and \eqref{eq:uniform bound for subgradient}, respectively. Under such setting, suppose that we apply the SubGM in~\Cref{alg:sgd for general problem} to solve~\eqref{eq:rms factorization}, where the initial point $\mU_0$ satisfies $\dist(\mU_0,\calU)< \frac{2\alpha}{\tau}$ and the geometrically diminishing step sizes $\mu_k = \rho^k\mu_0$ are used with $\mu_0$, $\rho$ satisfying~\eqref{eq:requirement on mu0}, \eqref{eq:requirement on rho}, respectively. Then, the sequence of iterates $\{\mU_k\}_{k\ge0}$ generated by the SubGM will converge to a point in $\calU$ at a linear rate:
	\[
	\dist(\mU_k,\calU) \leq \rho^k \max\left\{\dist(\mU_0,\calU), \mu_0\frac{\max\{\kappa^2,2\alpha^2\}}{\alpha} \right\}.
	\]
Moreover, the ground-truth matrix $\mX^\star$ can be exactly recovered by any point $\mU^\star \in \calU$ via $\mX^\star = \mU^\star \mU^{\star \T}$.
\end{thm}

We remark that a similar result for the smooth counterpart \eqref{eq:ms factorization} without any outliers is established in \cite[Theorem 3.3]{tu2015low}. Our \Cref{thm:rms} implies that the nonsmooth problem \eqref{eq:rms factorization} can be solved \emph{as efficiently as} its smooth counterpart \eqref{eq:ms factorization}, even in the presence of a substantial fraction of outliers in the measurement vector.

\subsection{{Initializing the SubGM}} \label{subsec:initialization} 
We now discuss some potential initialization strategies for the SubGM. A common approach to generating an appropriate initialization for matrix recovery-type problems is the spectral method. In our context, this entails simply computing the rank-$r$ approximation of $\tfrac{1}{m}\calA^*(\vy) = \tfrac{1}{m}\sum_{i=1}^m y_i \mA_i$, where $\calA^*$ is the adjoint operator of $\calA$. Specifically, let $\mP\mPi\mQ^\T$ be a rank-$r$ SVD of $\frac{1}{m}\calA^*(\vy) $, where $\mP, \mQ$ have orthonormal columns and $\mPi$ is an $r\times r$ diagonal matrix with the top $r$ singular values of $\frac{1}{m}\calA^*(\vy)$ along its diagonal. In the symmetric positive semidefinite case, we may assume without loss of generality that $\mA_1,\ldots,\mA_m$ are symmetric. Then, we can take $\mU_0 = \mP \mPi^{1/2}$ as the initialization. The main idea behind this approach is that when there is no outlier (i.e., $\vy = \calA(\mX^\star)$ as in \eqref{eq:ms model}), we have $\frac{1}{m}\calA^*(\vy) = \frac{1}{m}\calA^* (\calA(\mX^\star)) \approx \mX^\star$ when $\frac{1}{m}\calA^*\calA$ is close to an unitary operator for low-rank matrices. Thus, $\mU_0$ is also expected to be close to $\calU$. However, when the measurements are corrupted by outliers, it is possible that $\frac{1}{m}\calA^*(\vy)$ is perturbed away from $\frac{1}{m}\calA^*( \calA(\mX^\star))$ and thus $\mU_0$ may not be close enough to $\calU$. To mitigate the influence of outliers, Li et al. \cite{li2017nonconvex} have recently proposed a truncated spectral method for initialization, in which the spectral method is applied to an operator that is formed by using those measurements whose absolute values do not deviate too much from the median of the absolute values of certain sampled measurements; see \Cref{alg:intilization}. They showed that under appropriate conditions, the truncated spectral method can output an initialization that satisfies the requirement of \Cref{thm:rms}.

\begin{algorithm}
	\caption{Truncated Spectral Method for Initialization \cite{li2017nonconvex}}
	\mbox{{\bf Input:} measurement vector $\vy$; sensing matrices $\mA_1,\ldots,\mA_m$; threshold $\beta>0$;}
	\vspace{-0.9\baselineskip}
	\begin{algorithmic}[1]	
		\STATE set $\vy_1 = \{y_i\}_{i = 1}^{\lfloor m/2\rfloor}$, $\vy_2 =  \{y_i\}_{\lfloor m/2\rfloor+1 }^{m}$;
		\STATE Compute the rank-$r$ SVD of
		\[ \mE = \frac{1}{\lfloor m/2 \rfloor} \sum_{i=1}^{\lfloor m/2\rfloor}  y_i \mA_i \setI_{\{ |y_i|\leq \beta\cdot \text{median}(|\vy_2|)  \}} \]
		and denote it by $\mP\mPi\mQ^T$, where
		\[ \setI_{\{ |y_i|\leq \beta\cdot \text{median}(|\vy_2|)  \}} = \left\{
		\begin{array}{c@{\quad}l}
		1 & \mbox{if } |y_i|\leq \beta\cdot \text{median}(|\vy_2|), \\
		0 & \mbox{otherwise};
		\end{array}
		\right.
		\]
	\end{algorithmic}
	{\bf Output:}  $\mU_0 = \mP\mPi^{1/2}$, {$\mV_0 = \mQ\mPi^{1/2}$};
	\label{alg:intilization}
\end{algorithm}
{
\begin{thm}[proximity of initialization to optimal set: PSD case; cf.~{\cite[Theorem 3.3]{li2017nonconvex}}] \label{thm:initialization}
Let $r\ge1$ be given and set $\overline c = \tfrac{\|\mX^\star\|_F}{\sqrt{r} \sigma_r(\mX^\star)}$. Suppose that the matrices $\mA_1,\ldots,\mA_m \in \R^{n\times n}$ defining the linear measurement operator $\calA$ are symmetric and have i.i.d. standard Gaussian entries on and above the diagonal, and that the number of measurements $m$ satisfies $m \gtrsim \beta^2\overline{c}^2nr^2\log n$, where $\beta = 2\log\left( r^{1/4} \overline c^{1/2} + 20 \right)$. Furthermore, suppose that the fraction of outliers $p$ in the measurement vector $\vy$ satisfies $p \lesssim \frac{1}{\sqrt{r}\overline{c}}$. Then, with overwhelming probability, \Cref{alg:intilization} outputs an initialization $\mU_0 \in \R^{n\times r}$ satisfying $\dist(\mU_0, \calU) \lesssim \sigma_r(\mX^\star)$ and hence also the requirement of \Cref{thm:rms} (as $ \sigma_r(\mX^\star)$ is of the same order as $\frac{2\alpha}{\tau}$).
\end{thm}
}

Note that the requirements on the number of measurements and the fraction of outliers that can be tolerated are slightly more stringent than those in \Cref{thm: L1-2 RIP for Gaussian} and \Cref{thm:rms}.  However, as will be illustrated in \Cref{sec:experiments}, our numerical experiments show that even a randomly initialized SubGM can very efficiently find the global minimum and hence recover the ground-truth matrix $\mX^\star$. A theoretical justification of such a phenomenon will be the subject of a future study. We suspect that it may be possible to relax the requirement on the initialization in \Cref{thm:rms} or to show that the SubGM enters the region $\left\{ \mU : \dist(\mU, \calU) < \frac{2\alpha}{\tau} \right\}$ very quickly even though the random initialization lies outside of this region.

\section{Nonconvex Robust Low-Rank Matrix Recovery: General Case}
\label{sec:nonsymmetric}

In this section we consider the general setting where $\mX^\star$ is a rank-$r$ $n_1 \times n_2$ matrix. To extend the nonsmooth nonconvex formulation~\eqref{eq:rms factorization} to this setting, a natural approach is to use the factorization $\mX =\mU\mV^\T$ with $\mU\in\R^{n_1\times r}$ and $\mV\in\R^{n_2\times r}$. However, such a factorization is ambiguous in the sense that if $\mX = \mU\mV^\T$, then $\mX = (\mU\mT)(\mV \mT^{-\T})^\T$ for any invertible matrix $\mT \in \R^{r\times r}$. To address this issue, we introduce the nonsmooth nonconvex regularizer
\e\label{eq:nonsmooth regularizer}
\phi(\mU,\mV):=\|\mU^\T\mU - \mV^\T\mV\|_F,
\ee
which aims to balance the factors $\mU$ and $\mV$, and solve the following regularized problem:
\e
\minimize_{\mU\in\R^{n_1\times r},\mV\in\R^{n_2\times r}} \left\{ g(\mU,\mV):= \frac{1}{m}\|\vy - \calA(\mU\mV^\T) \|_1 + \lambda \|\mU^\T\mU - \mV^\T\mV\|_F \right\}.
\label{eq:rms-nonsymmetric}\ee
Here, $\lambda>0$ is a regularization parameter. We remark that a similar regularizer, namely,
\[ \widetilde \phi(\mU,\mV):=\|\mU^\T\mU - \mV^\T\mV\|_F^2, \]
has been introduced in \cite{tu2015low,park2016non,zhu2018global} to account for the ambiguities caused by invertible transformations when minimizing the squared $\ell_2$-loss function $(\mU,\mV) \mapsto \frac{1}{m} \|\vy - \calA(\mU\mV^\T)\|_2^2$. However, such a regularizer is not entirely suitable for the $\ell_1$-loss function, as it is no longer clear that the resulting problem will satisfy the sharpness condition in~\Cref{def:shaprness}.

To simplify notation, we stack $\mU$ and $\mV$ together as $\mW = \begin{bmatrix} \mU \\ \mV \end{bmatrix}$ and write $g(\mW)$ for $g(\mU,\mV)$. Observe that the regularizer $\phi$ achieves its minimum value of $0$ when $\mU$ and $\mV$ have the same Gram matrices; i.e., $ \mU^\T\mU = \mV^\T \mV$. Now, let $\mX^\star = \mPhi\mSigma\mPsi^\T$ be a rank-$r$ SVD of $\mX^\star$, where $\mPhi\in\R^{n_1\times r}, \mPsi\in\R^{n_2\times r}$ have orthonormal columns and $\mSigma \in \R^{r\times r}$ is a diagonal matrix.  Define
\[
\mU^\star = \mPhi\mSigma^{1/2},\quad \mV^\star = \mPsi\mSigma^{1/2}, \quad \mW^\star = \begin{bmatrix} \mU^\star \\ \mV^\star \end{bmatrix}.
\] 
The orthogonal invariance of $g$ (i.e., $g(\mW) = g(\mW\mR)$ for any $\mR \in \calO_r$) implies that $g$ is constant on the set
\[
\calW:=\left\{\mW^\star\mR: \mR\in\calO_r \right\}.
\]

\subsection{{Sharpness and Exact Recovery}} \label{subsec:general sharp}
Our immediate goal is to show that $\calW$ is the set of global minima of~\eqref{eq:rms-nonsymmetric}. Towards that end, let $0<\delta<\frac{1}{3}\sqrt{\frac{2}{\pi}}$ be given. Suppose that the fraction of outliers $p$ in the measurement vector $\vy$ satisfies~\eqref{eq:p and delta}, and that the linear operators $\calA:\R^{n_1\times n_2}\rightarrow\R^m$ and $\calA_{\Omega^c}:\R^{n_1\times n_2} \rightarrow \R^{|\Omega^c|}$ possess the $\ell_1/\ell_2$-RIP~\eqref{eq:L1-2 RIP for Gaussian} and~\eqref{eq:RIP partial A}, respectively.\footnote{It can be shown that modulo the constants, the Gaussian measurement operator $\calA:\R^{n_1\times n_2} \rightarrow \R^m$ will possess the $\ell_1/\ell_2$-RIPs~\eqref{eq:L1-2 RIP for Gaussian} and~\eqref{eq:RIP partial A} with high probability as long as $m \gtrsim \max\{n_1,n_2\}r$. To avoid any distraction caused by the new constants, we shall simply use the $\ell_1/\ell_2$-RIPs~\eqref{eq:L1-2 RIP for Gaussian} and~\eqref{eq:RIP partial A} in our derivation.} Using the argument in the proof of \Cref{lem:sharpness rms}, we get

\e
\overline g(\mW) -\overline g(\mW^\star) \geq \left(2(1-p) \left( \sqrt{\frac{2}{\pi}} - \delta \right) - \left(\sqrt{\frac{2}{\pi}} + \delta\right)\right)\|  \mU\mV^\T - \mX^\star \|_F,
\label{eq:sharpness rms first part}
\ee
where
\[ \overline g(\mW) = \frac{1}{m} \|\vy - \calA(\mU\mV^\T) \|_1. \]
In particular, we see that $\overline g(\mW) >\overline g(\mW^\star)$ whenever $\mU\mV^\T \neq \mX^\star$. Since $\mU^{\star\T}\mU^\star = \mV^{\star\T}\mV^\star$ by construction, we conclude that $\mW^\star$ is a global minimum of \eqref{eq:rms-nonsymmetric}, as $\mW^\star$ is a global minimum of both the first term $\overline g$ and the second term $\phi$ of $g$. It then follows from the orthogonal invariance of $g$ that every element in $\calW$ is a global minimum of~\eqref{eq:rms-nonsymmetric}. The following result further establishes that $\calW$ is exactly the set of global minima of~\eqref{eq:rms-nonsymmetric} and $g$ is sharp.

\begin{prop}[sharpness and exact recovery with outliers: general case] 
Let $0<\delta<\frac{1}{3}\sqrt{\frac{2}{\pi}}$ be given. Suppose that the fraction of outliers $p$ satsifies~\eqref{eq:p and delta}, and that the linear operators $\calA$ and $\calA_{\Omega^c}$ possess the $\ell_1/\ell_2$-RIP~\eqref{eq:L1-2 RIP for Gaussian} and~\eqref{eq:RIP partial A}, respectively. Then, the objective function $g$ in \eqref{eq:rms-nonsymmetric} satisfies
\[
	g(\mW) - g(\mW^\star) \geq \alpha  \dist(\mW,\calW)
\]
for any $\mW \in \R^{(n_1+n_2)\times r}$, where
	\e
	\alpha = \sqrt{\sqrt{2}-1}\cdot \min\left\{ 2(1-p)\left(\sqrt{\frac{2}{\pi}} - \delta\right) - \left(\sqrt{\frac{2}{\pi}} + \delta\right), 2\lambda \right\}\cdot \sigma_r(\mX^\star) >0.
	\label{eq:sharpness parameter nonsymmetric}\ee
	In particular, the set $\calW$ is precisely the set of global minima of \eqref{eq:rms-nonsymmetric} and the objective function $g$ is sharp with parameter $\alpha>0$.
	\label{lem:sharpness nonsymmetric}
\end{prop}
\begin{proof}[Proof of \Cref{lem:sharpness nonsymmetric}]
	Let $\zeta(p,\delta) = 2(1-p)\left( \sqrt{\frac{2}{\pi}} - \delta\right) - \left( \sqrt{\frac{2}{\pi}} + \delta\right)$. Since $\mU^{\star\T}\mU^{\star} = \mV^{\star\T}\mV^{\star}$, we have $\phi(\mW^\star) = 0$ by~\eqref{eq:nonsmooth regularizer}  and
\begin{align*}
		&g(\mW) - g(\mW^\star) = \frac{1}{m}\|\vy - \calA(\mU\mV^\T) \|_1 - \frac{1}{m}\|\vy - \calA(\mX^\star) \|_1 + \lambda \|\mU^\T\mU - \mV^\T\mV\|_F \\
		&\geq \zeta(p,\delta)\| \mX^\star  - \mU\mV^\T \|_F + \lambda \|\mU^\T\mU - \mV^\T\mV\|_F\\
		&  \geq \min\left\{ \zeta(p,\delta), 2\lambda \right\} \left( \| \mX^\star  - \mU\mV^\T \|_F + \frac{1}{2}\|\mU^\T\mU - \mV^\T\mV\|_F\right)\\
		& \geq \min\left\{ \zeta(p,\delta), 2\lambda \right\} \sqrt{\| \mX^\star  - \mU\mV^\T \|_F^2 +  \frac{1}{4}\|\mU^\T\mU - \mV^\T\mV\|_F^2}\\
		&\geq \min\left\{ \frac{\zeta(p,\delta)}{2}, \lambda \right\} \|\mW\mW^\T - \mW^{\star}\mW^{\star\T}\|_F\\
		& \geq \min\left\{ \frac{\zeta(p,\delta)}{2}, \lambda \right\} \sqrt{2\left(\sqrt{2}-1\right)}\sigma_r(\mW^\star)\dist(\mW,\calW)\\
		&=  \min\left\{ \zeta(p,\delta), 2\lambda \right\} \sqrt{\sqrt{2}-1}\sigma_r^{1/2}(\mX^\star)\dist(\mW,\calW),
\end{align*}
where the first inequality follws from~\eqref{eq:sharpness rms first part}, the fourth inequality follows from
\begin{align*}
& \|\mX^\star - \mU\mV^\T \|_F^2 + \frac{1}{4} \|\mU^\T\mU - \mV^\T \mV\|_F^2 = \| \mU^\star\mV^{\star\T} - \mU\mV^\T \|_F^2 + \frac{1}{4} \|\mU^\T\mU - \mV^\T \mV\|_F^2 \\
&= \frac{1}{4}\|\mW\mW^\T - \mW^\star\mW^{\star\T}\|_F^2 + \nu(\mW)
\end{align*}
with
\begin{align*}
	\nu(\mW) &= \frac{1}{2}\|\mU\mV^\T - \mU^\star\mV^{\star\T}\|_F^2  + \frac{1}{4} \|\mU^\T\mU - \mV^\T \mV\|_F^2\\
	&\quad  - \frac{1}{4}\|\mU\mU^\T - \mU^\star\mU^{\star\T}\|_F^2 - \frac{1}{4}\|\mV\mV^\T - \mV^\star\mV^{\star\T}\|_F^2\\
	&= {\frac{1}{2}\|\mU^\T\mU^\star\|_F^2  + \frac{1}{2}\|\mV^\T\mV^\star\|_F^2 - \left\langle \mU\mV^\T, \mU^\star\mV^{\star\T} \right\rangle} \\
	&\quad { +\frac{1}{2}\|\mU^\star\mV^{\star\T}\|_F^2 -\frac{1}{4}\|\mU^\star\mU^{\star\T}\|_F^2 - \frac{1}{4}\|\mV^\star\mV^{\star\T}\|_F^2}\\
	& = \frac{1}{2}\|\mU^\T\mU^\star - \mV^\T\mV^\star\|_F^2 + \frac{1}{2}\|\mU^\star\mV^{\star\T}\|_F^2 -\frac{1}{4}\|\mU^\star\mU^{\star\T}\|_F^2 - \frac{1}{4}\|\mV^\star\mV^{\star\T}\|_F^2\\
	& = \frac{1}{2}\|\mU^\T\mU^\star - \mV^\T\mV^\star\|_F^2 \geq 0
	\end{align*}
 {(recall that $\mU^{\starT} \mU^\star = \mV^{\starT} \mV^\star$)}, the fifth inequality is from \Cref{lem:Procrustes problem}, and the last equality follows from the fact that $\sigma_r(\mW^\star) = \sqrt{2}\sigma_r^{1/2}(\mX^\star)$. This completes the proof.
\end{proof}		

By comparing \Cref{lem:sharpness rms} and \Cref{lem:sharpness nonsymmetric}, we see that the fraction of outliers that can be tolerated for exact recovery is the same in both the symmetric positive semidefinite and general cases. Moreover, the sharpness parameter $\alpha$ in~\eqref{eq:sharpness parameter nonsymmetric} demonstrates the role that the regularizer $\phi$ plays: When the regularizer $\phi$ is absent (which corresponds to $\lambda = 0$), although every element in $\calW$ is still a global minimum of~\eqref{eq:rms-nonsymmetric}, we cannot guarantee that there is no other global minimum. Indeed, when $\lambda = 0$, the pair $(\mU^\star \mT,\mV^\star \mT^{-\T})$ is a global minimum of~\eqref{eq:rms-nonsymmetric} for any invertible matrix $\mT\in\R^{r\times r}$. However, when $\lambda>0$, the regularizer $\phi$ ensures that the pair $(\mU^\star \mT,\mV^\star \mT^{-\T})$ is a global minimum of~\eqref{eq:rms-nonsymmetric} only when $\mT\in\calO_r$.

\subsection{Weak Convexity} \label{subsec:general weak-cvx}
Let us now establish the weak convexity of the objective function $g$ in \eqref{eq:rms-nonsymmetric}.
\begin{prop}[weak convexity: general case]\label{lem:weak convex nonsymmetric}
	Suppose that the measurement operator $\calA$ satisfies the $\ell_1/\ell_2$-RIP \eqref{eq:L1-2 RIP for Gaussian}. Then, the objective function $g$ in \eqref{eq:rms-nonsymmetric} is weakly convex with parameter
	\e
	\tau = {\sqrt{\frac{2}{\pi}} + \delta} + 2\lambda.
	\label{eq:weak convexity parameter nonsymmetric}	\ee
\end{prop}

\begin{proof}[Proof of \Cref{lem:weak convex nonsymmetric}]
Since $g=\overline{g}+\lambda\phi$, it suffices to show that $\overline{g}$ and $\phi$ are both weakly convex. Similar to~\eqref{eq:subdifferential f}, we apply the chain rule for subdifferentials~\cite[Corollary 8.11 and Theorem 10.6]{RW04} to get
\[ \partial\overline{g}(\mW) = \frac{1}{m}
\begin{bmatrix}
\calA^*\left( \Sign\left( \calA(\mU\mV^\T)-\vy \right)\right) \mV \\
\left( \calA^*\left( \Sign\left( \calA(\mU\mV^\T)-\vy \right)\right) \right)^\T \mU
\end{bmatrix}. 
\]
Using this and the argument in the proof of \Cref{prop:no outlier weak convex}, we can show that for any $\mW,\mW' \in \R^{(n_1+n_2)\times r}$,
	\begin{align*}
		\overline g(\mW') &\geq \overline g(\mW) + \left\langle \mD,\mW' - \mW \right\rangle - \left( \sqrt{\frac{2}{\pi}} + \delta\right) \|(\mU' - \mU)(\mV' - \mV)^\T \|_F \\
		&\ge \overline g(\mW) + \left\langle \mD,\mW' - \mW \right\rangle - \left(\frac{\sqrt{{2}/{\pi}} + \delta}{{2}}\right) \| \mW'-\mW \|_F^2 , \ \forall \ \mD \in \partial \overline g(\mW);
	\end{align*}
	i.e., the function $\overline{g}$ is weakly convex with parameter $\tau_{\overline{g}}=\sqrt{\frac{2}{\pi}} + \delta$.
	
	Next, define the matrices
	\[
	\underline \mW =  \begin{bmatrix} \mU \\ -\mV \end{bmatrix}, \quad \underline \mW' =  \begin{bmatrix} \mU' \\ -\mV' \end{bmatrix}
	\]
	and note that $\underline \mW^\T \mW = \mU^\T\mU - \mV^\T\mV$. Furthermore, define the function $\psi:\R^{r\times r} \rightarrow \R$ by
	\[
	\psi(\mC) = \|\mC\|_F,
	\]
	whose subdifferential is
	\[
	\partial \psi(\mC) = \left\{\begin{matrix}  \left\{ \frac{\mC}{\|\mC\|_F} \right\}, & \mC \neq \mzero,\\ \left\{ \mB\in \R^{r\times r} : \|\mB\|_F \le 1  \right\},   & \mC = \mzero.        \end{matrix}    \right.
	\]
	Upon setting $\mDelta = \mW' - \mW$ and $\underline\mDelta = \underline\mW' - \underline\mW$, we compute
	\e\begin{split}
		&\phi(\mW')=\|\underline\mW '^\T \mW'\|_F \\
		&=
		\|\underline\mW^\T\mW +\underline \mW^\T \mDelta +  \underline\mDelta^\T \mW + \underline\mDelta^\T \mDelta\|_F \\
		& \geq \|\underline\mW^\T\mW +\underline \mW^\T  \mDelta +  \underline\mDelta^\T \mW\|_F- \| \underline\mDelta^\T \mDelta\|_F \\
		& \geq  \|\underline\mW^\T\mW\|_F + \left\langle  \mPsi, \underline \mW^\T  \mDelta +  \underline \mDelta^\T \mW\right\rangle - \|\underline\mDelta^\T \mDelta\|_F,
	\end{split}
	\label{eq:weak convex rho}\ee
	where the last inequality holds for any $\mPsi\in \partial \psi(\underline\mW^\T\mW)$ due to the convexity of the Frobenius norm. Since the Frobenius norm is subdifferentially regular~\cite[Example 7.27]{RW04}, the chain rule for subdifferentials~\cite[Corollary 8.11 and Theorem 10.6]{RW04} yields
\e
\partial \phi(\mW) = \left\{ \underline \mW (\mPsi+\mPsi^\T) :  \mPsi \in \partial \psi(\underline\mW^\T\mW) \right\}.
\label{eq:reg-subdiff}
\ee
	It follows from~\eqref{eq:weak convex rho} and~\eqref{eq:reg-subdiff} that
	\begin{align*}
		\phi(\mW') &\geq \phi(\mW) + \left\langle  \mPhi , \mW' - \mW\right\rangle  - \|\underline\mDelta^\T\mDelta\|_F \\
		&\ge \phi(\mW) + \left\langle  \mPhi , \mW' - \mW\right\rangle - \| \mW'-\mW \|_F^2, \ \forall \mPhi \in \partial \phi(\mW);
	\end{align*}
	i.e., the function $\phi$ is weakly convex with parameter $\tau_\phi = 2$.
	
	Putting the above results together, we conclude that $g = \overline{g}+\lambda\phi$ is weakly convex with parameter $\tau = \tau_{\overline{g}} + \lambda\tau_\phi$, as desired.
\end{proof}

Unlike the sharpness condition in \Cref{lem:sharpness nonsymmetric} that requires $\lambda>0$, the weak convexity condition in~\Cref{lem:weak convex nonsymmetric} holds even when $\lambda = 0$. Although the parameters $\alpha$ and $\tau$ in \eqref{eq:sharpness parameter nonsymmetric} and \eqref{eq:weak convexity parameter nonsymmetric} increase as $\lambda$ increases from $0$, the former becomes constant when $\lambda \geq \frac{2(1-p) \left( \sqrt{{2}/{\pi}} - \delta \right) - \left(\sqrt{{2}/{\pi}} + \delta\right)}{2}$. In view of \Cref{thm:linear convergence of SubGM}, it is desirable to choose $\lambda$ so that the local linear convergence region $\left\{ \vx : \dist(\vx,\calX) < \frac{2\alpha}{\tau} \right\}$ of the SubGM is as large as possible. Such consideration suggests that we should set
\[
\lambda = \frac{2(1-p)\left( \sqrt{{2}/{\pi}} - \delta \right) - \left( \sqrt{{2}/{\pi}} + \delta \right)}{2}.
\]

\subsection{{Putting Everything Together}} 
{As in \Cref{subsec:summarize symmetric guarantees}, before we can invoke~\Cref{thm:linear convergence of SubGM} to establish convergence guarantees for the SubGM when applied to the general robust low-rank matrix recovery problem \eqref{eq:rms-nonsymmetric}, we need to bound the norm of any subgradient of $g$ in a neighborhood of~$\calW$. This is achieved by the following result:}

\begin{prop}[bound on subgradient norm: general case] 
{Suppose that the measurement operator $\calA$ satisfies the $\ell_1/\ell_2$-RIP \eqref{eq:L1-2 RIP for Gaussian}. Then, for any $\mW\in\R^{(n_1 + n_2)\times r}$ satisfying $\dist(\mW,\calW) \leq \frac{2\alpha}{\tau}$, we have} \label{lem:nonsymmetric bound kappa} 
	\begin{align}
	{ \|\mD \|_F \leq \kappa = \max\left\{\sqrt{\frac{2}{\pi}}+ \delta, \lambda \right\} \left( \|\mW^\star\|_F + \frac{2\alpha}{\tau} \right), \ \forall \ \mD\in \partial g(\mW). }
	\label{eq:kappa-gen}
	\end{align}
\end{prop}

\begin{proof}[Proof of \Cref{lem:nonsymmetric bound kappa}] 
Observe that for any $\mW, \mW'\in\R^{(n_1 + n_2)\times r}$,
\begin{align*}
&\left|g(\mW')- g(\mW)\right| \leq \left|\overline g(\mW')- \overline g(\mW)\right| + \lambda\left|\phi(\mW')- \phi(\mW)\right| \\
& \leq \frac{1}{m} \left\| \calA(\mU\mV^\T - \mU'\mV'^\T) \right\|_1 +\lambda\left( \left\| \mU^\T\mU - \mU'^\T\mU' \right\|_F + \left\|\mV^\T\mV - \mV'^\T\mV' \right\|_F \right) \\
& \leq \left(\sqrt{\frac{2}{\pi}}+ \delta \right) \left\| \mU\mV^\T - \mU'\mV'^\T \right\|_F + \lambda\left( \left\| \mU^\T\mU - \mU'^\T\mU' \right\|_F + \left\| \mV^\T\mV - \mV'^\T\mV' \right\|_F \right)\\
& \leq \left( \sqrt{\frac{2}{\pi}}+ \delta \right) \left(\|\mV\|_F \|\mU - \mU'\|_F + \|\mU'\|_F \|\mV - \mV'\|_F \right)\\
& \quad + \lambda\left(\|\mU\|_F + \|\mU'\|_F  \right)\|\mU - \mU'\|_F
 + \lambda \left(\|\mV\|_F + \|\mV'\|_F  \right)\|\mV - \mV'\|_F\\
& \leq \max\left\{ \sqrt{\frac{2}{\pi}}+ \delta, \lambda \right\} \left(\|\mW\|_F + \|\mW'\|_F \right) \| \mW - \mW'\|_F ,
\end{align*}
where the third inequality follows from the $\ell_1/\ell_2$-RIP~\eqref{eq:L1-2 RIP for Gaussian}. Thus, similar to the derivation of~\eqref{eq:uniform bound for subgradient}, for any $\mW\in\R^{(n_1+n_2)\times r}$ satisfying $\dist(\mW,\calW)\leq \frac{2\alpha}{\tau}$, where $\alpha$ and $\tau$ are given in~\eqref{eq:sharpness parameter nonsymmetric} and \eqref{eq:weak convexity parameter nonsymmetric}, respectively, we have
\begin{align*}
\|\mD \|_F &\leq  \max\left\{ \sqrt{\frac{2}{\pi}}+ \delta, \lambda \right\}\|\mW\|_F \\
&\leq \max\left\{\sqrt{\frac{2}{\pi}}+ \delta, \lambda \right\} \left( \|\mW^\star\|_F + \frac{2\alpha}{\tau} \right), \ \forall \ \mD\in \partial g(\mW).
\end{align*}
\end{proof}

By collecting \Cref{lem:sharpness nonsymmetric}, \Cref{lem:weak convex nonsymmetric}, and \Cref{lem:nonsymmetric bound kappa} together and invoking \Cref{thm:linear convergence of SubGM}, 
{we obtain the following guarantees when the SubGM is used to solve the general robust low-rank matrix recovery problem \eqref{eq:rms-nonsymmetric}:}


{
\begin{thm}[nonconvex robust low-rank matrix recovery: general case] 
	Consider the measurement model~\eqref{eq:rms model}, where $\mX^\star$ is an $n_1\times n_2$ rank-$r$  matrix. Let $0<\delta<\frac{1}{3}\sqrt{\frac{2}{\pi}}$ be given. Suppose that the fraction of outliers $p$ in the measurement vector $\vy$ satisfies~\eqref{eq:p and delta}, and that the linear operators $\calA$, $\calA_{\Omega^c}$ possess the $\ell_1/\ell_2$-RIP~\eqref{eq:L1-2 RIP for Gaussian},~\eqref{eq:RIP partial A}, respectively. Let $\alpha$, $\tau$, and $\kappa$ be given by~\eqref{eq:sharpness parameter nonsymmetric}, \eqref{eq:weak convexity parameter nonsymmetric}, and~\eqref{eq:kappa-gen}, respectively. Under such setting, suppose that we apply the SubGM in~\Cref{alg:sgd for general problem} to solve~\eqref{eq:rms-nonsymmetric}, where the initial point $\mW_0$ satisfies $\dist(\mW_0,\calW)< \frac{2\alpha}{\tau}$ and the geometrically diminishing step sizes $\mu_k = \rho^k\mu_0$ are used with $\mu_0$, $\rho$ satisfying~\eqref{eq:requirement on mu0}, \eqref{eq:requirement on rho}, respectively. Then, the sequence of iterates $\{\mW_k\}_{k\ge0}$ generated by the SubGM will converge to a point in $\calW$ at a linear rate:
	\[
	\dist(\mW_k,\calW) \leq \rho^k \max\left\{\dist(\mW_0,\calW), \mu_0\frac{\max\{\kappa^2,2\alpha^2\}}{\alpha} \right\}.
	\]
	Moreover, the ground-truth matrix $\mX^\star$ can be exactly recovered by any point $\mW^\star \in \calW$ via $\mX^\star = \mU^\star \mV^{\star \T}$.
	\label{thm:general rms}
\end{thm}
}

\subsection{{Initializing the SubGM}} 
{In the general case, we can still use the truncated spectral method in \Cref{alg:intilization} to obtain a good initialization for the SubGM. Specifically, we take $\mW_0 = \begin{bmatrix}\mU_0^\T &\mV_0^\T \end{bmatrix}^\T$ as the initialization, where $\mU_0,\mV_0$ are the outputs of~\Cref{alg:intilization}. Then, we have the following result, which is essentially a restatement of~\cite[Theorem 3.3]{li2017nonconvex}:
\begin{thm}[proximity of initialization to optimal set: general case]
Let $r\ge1$ be given and set $n=n_1+n_2$, $\overline c = \tfrac{\|\mX^\star\|_F}{\sqrt{r} \sigma_r(\mX^\star)}$. Suppose that the matrices $\mA_1,\ldots,\mA_m \in \R^{n_1\times n_2}$ defining the linear measurement operator $\calA$ have i.i.d. standard Gaussian entries, and that the number of measurements $m$ satisfies $m \gtrsim \beta^2\overline{c}^2nr^2\log n$, where $\beta = 2\log\left( r^{1/4} \overline c^{1/2} + 20 \right)$. Furthermore, suppose that the fraction of outliers $p$ in the measurement vector $\vy$ satisfies $p \lesssim \frac{1}{\sqrt{r}\overline{c}}$. Then, with overwhelming probability, \Cref{alg:intilization} outputs an initialization $\mW_0 \in \R^{(n_1+n_2)\times r}$ satisfying $\dist(\mW_0, \calU) \lesssim \sigma_r(\mX^\star)$ and hence also the requirement of \Cref{thm:general rms}.
\end{thm}
 }

\section{Experiments\label{sec:experiments}}

In this section we conduct experiments to illustrate the performance of the SubGM when applied to robust low-rank matrix recovery problems. The experiments on synthetic data show that the SubGM can exactly and efficiently recover the underlying low-rank matrix from its linear measurements even in the presence of outliers, thus corroborating the result in~\Cref{thm:rms}.

We generate the underlying low-rank matrix $\mX^\star = \mU^\star\mU^{\star \T}$ by generating $\mU^\star\in\R^{n\times r}$ with \emph{i.i.d.}  standard Gaussian entries. Similarly, we generate the entries of the $m$ sensing matrices $\mA_1,\ldots,\mA_m\in\R^{n\times n}$ (which define the linear measurement operator $\calA$) in an \emph{i.i.d.} fashion according to the standard Gaussian distribution. To generate the outlier vector $\vs^\star\in\R^m$, we first randomly select $pm$ locations. Then, we fill each of the selected location with an \emph{i.i.d.} mean 0 and variance 100 Gaussian entry, while the remaining locations are set to 0.  Here, $p$ is the ratio of the nonzero elements in $\vs^\star$. According to \eqref{eq:rms model}, the measurement vector $\vy$ is then generated by $\vy = \calA(\mX^\star)+ \vs^\star$; i.e.,  $y_i = \langle \mA_i, \mX^\star\rangle + s^\star_i$ for $i=1,\ldots,m$.

To illustrate the performance of the SubGM for recovering the underlying low-rank matrix $\mX^\star$ from $\vy$, we first set $n = 50$, $r = 5$, and $p = 0.3$. Throughout the experiments, we initialize the SubGM with a randomly generated standard Gaussian vector, as it gives similar practical performance as the one obtained by the truncated spectral method in \Cref{alg:intilization}. {We first run the SubGM for $10^4$ iterations using the geometrically diminishing step sizes $\mu_k = \rho^k\mu_0$, where the initial step size $\mu_0$ and decay rate $\rho$ are selected from $\{0.1,0.5,1,10\}$ and $\{0.80,0.81,0.82,\ldots,0.99\}$, respectively.
For each pair of parameters $(\mu_0,\rho)$, we plot the distance of the last iterate to $\calU$ (i.e., $\dist(\mU_{10^4}, \calU)$) in~\Cref{fig: rms-grid}. When the SubGM diverges, we simply set $\dist(\mU_{10^4}, \calU) = 10^4$ for the purpose of presenting all results in the same figure. As observed from \Cref{fig: rms-grid}, the SubGM diverges when $\mu_0$ is large, say, $\mu_0 = 10$. On the other hand, it converges to a global minimum when $\mu_0=1$, $\rho \in [0.93, 0.99]$ and $\mu_0=0.5$, $\rho \in [0.95, 0.99]$. It is worth noting that the SubGM converges to a global minimum when $\mu_0 =1, \rho = 0.93$, but not when $\mu_0 =0.5, \rho = 0.93$. This is consistent with~\Cref{thm:linear convergence of SubGM}, which shows that a larger initial step size $\mu_0$ allows for a smaller decay rate $\rho$. Such a phenomenon can also be observed in the case where $\mu_0 = 0.1$, for which the SubGM fails to find a global minimum even when $\rho \in [0.95, 0.99]$. 
  
In \Cref{fig: rms-gd}, we fix $\mu_0 = 1$ and plot the convergence behavior of the SubGM with $\rho \in\{0.9, 0.93,0.96,0.99\}$. As observed from the figure, when $\rho$ is not too small (say, larger than $0.93$), the distances $\{\dist(\mU_k,\calU)\}_{k\ge0}$ converge to $0$ at a linear rate, thus implying that the SubGM with geometrically diminishing step sizes can exactly recover the underlying low-rank matrix $\mX^\star$. We observe that a smaller $\rho$ gives faster convergence. This corroborates the results in~\Cref{thm:linear convergence of SubGM}, which guarantee that $\{\dist(\mU_k,\calU)\}_{k\ge0}$ decays at the rate $O(\rho^k)$ as long as $\rho$ is not too small (i.e., satisfying \eqref{eq:requirement on rho}). We also consider the SubGM with the Polyak step size rule~\cite{P69}, which, in the context of~\eqref{eq:rms factorization}, is given by $\mu_k = \tfrac{f(\mU_k) - f^\star}{\|\vd_k\|^2}$, where $f^\star$ is the optimal value of~\eqref{eq:rms factorization} and $\vd_k\in \partial f(\mU_k)$ (the method terminates when $\vd_k = \vzero$). The convergence rate of such method for sharp weakly convex minimization has been analyzed in~\cite{davis2018subgradient}. We plot the convergence behavior of the SubGM with the Polyak step size rule in~\Cref{fig: rms-gd}, which also shows its linear convergence. However, we note that the Polyak step size rule is generally not easy to implement, as it requires the knowledge of $f^\star$.
 
Then, we consider the SubGM with piecewise geometrically diminishing step sizes, which dates as far back as to the work~\cite{SS68} and has recently been used in~\cite{Zhu18DPCP}. Specifically, we set $\mu_k = \frac{1}{2^{\lfloor k/N\rfloor}}$ with $N\in \{50, 100, 200\}$. Compared to the vanilla strategy~\eqref{eq:linear step size}, the piecewise strategy allows for a smaller decay rate $\rho$ (here, we use $\rho=\tfrac{1}{2}$) and keeps the same step size for $N$ iterations.
As can be seen from \Cref{fig: rms-pgd}, the method converges at a piecewise linear rate. Nevertheless, we observe that the piecewise strategy is slightly less efficient than the vanilla one in general. 

We also consider a modified backtracking line search strategy in \cite{nocedal2006numerical} to choose the step size. Although such a strategy is generally designed for smooth problems, it is empirically used in \cite{Zhu18DPCP} for a nonsmooth nonconvex optimization problem to achieve fast convergence.  Inspired by the strategy of choosing geometrically diminishing step sizes, we modify the backtracking line search strategy in \cite{nocedal2006numerical} by (i) setting $\mu_k = \mu_{k-1}$ and (ii) reducing it according to $\mu_k \leftarrow \mu_k \rho$ until the condition $f(\mU_k - \mu_k \vd_k) > f(\mU_k) - \eta \mu_k \|\vd_k\|$ is satisfied. We set $\eta = 10^{-3}$, $\rho = 0.85$, $\mu_0=1$ and plot the convergence behavior of the resulting method in \Cref{fig: rms-line}. As can be seen from the figure, the method converges at a linear rate. Moreover, we observe empirically that the choice of parameters above works for other settings (i.e., different $n, r, m, p$). We leave the convergence analysis of the SubGM with backtracking line search as a future work.
}

\begin{figure}[!htp]
		\begin{subfigure}{0.49\linewidth}
			\centerline{
				\includegraphics[width=1\linewidth]{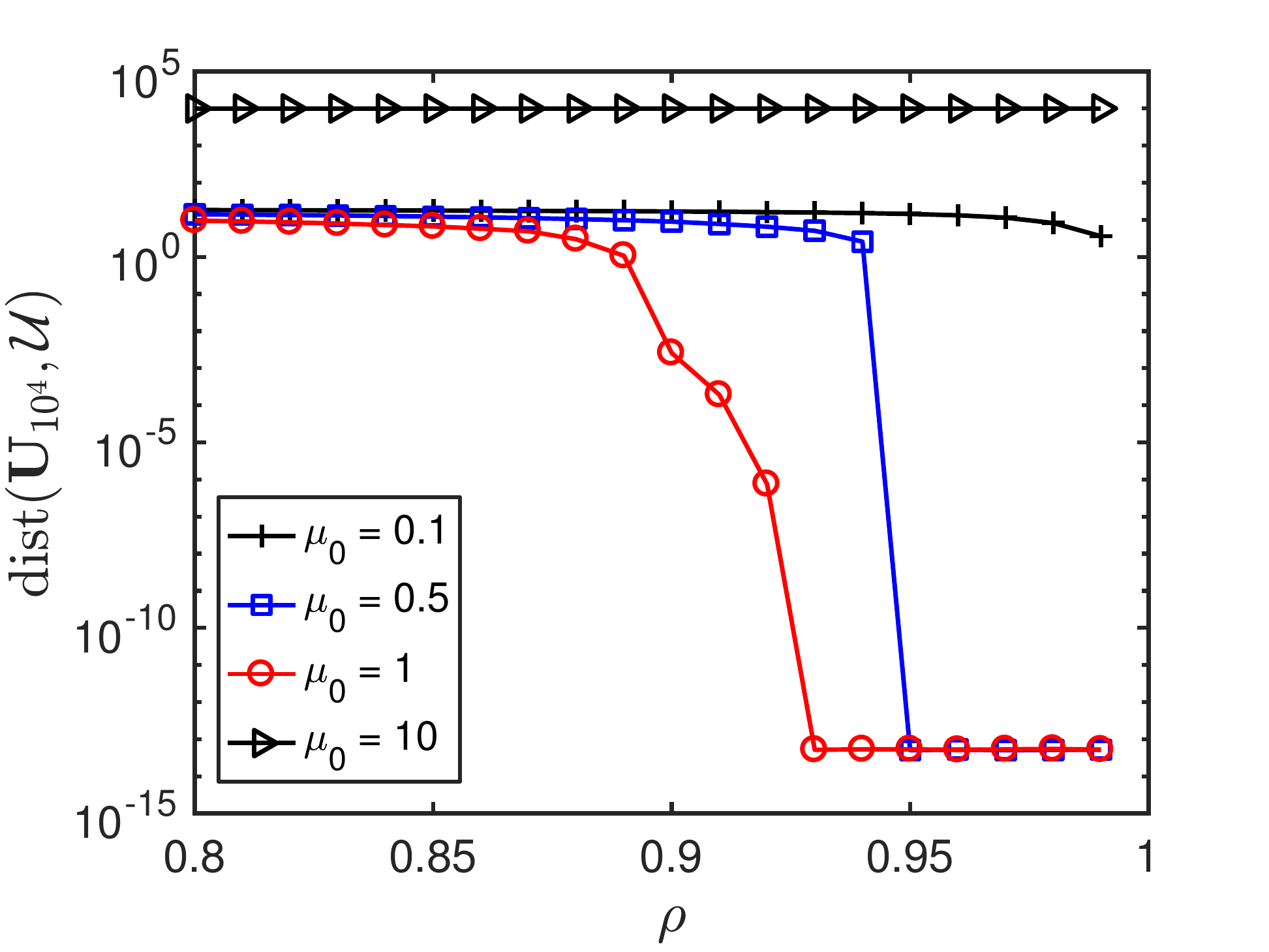}}
			\caption{Distance of last iterate to optimal set with $\mu_0 \in \{0.1,0.5,1,10\}$ and $\rho \in \{0.80,0.81,\ldots,0.99\}$}
			\label{fig: rms-grid}
		\end{subfigure}
		\begin{subfigure}{0.49\linewidth}
			\centerline{
				\includegraphics[width=1\linewidth]{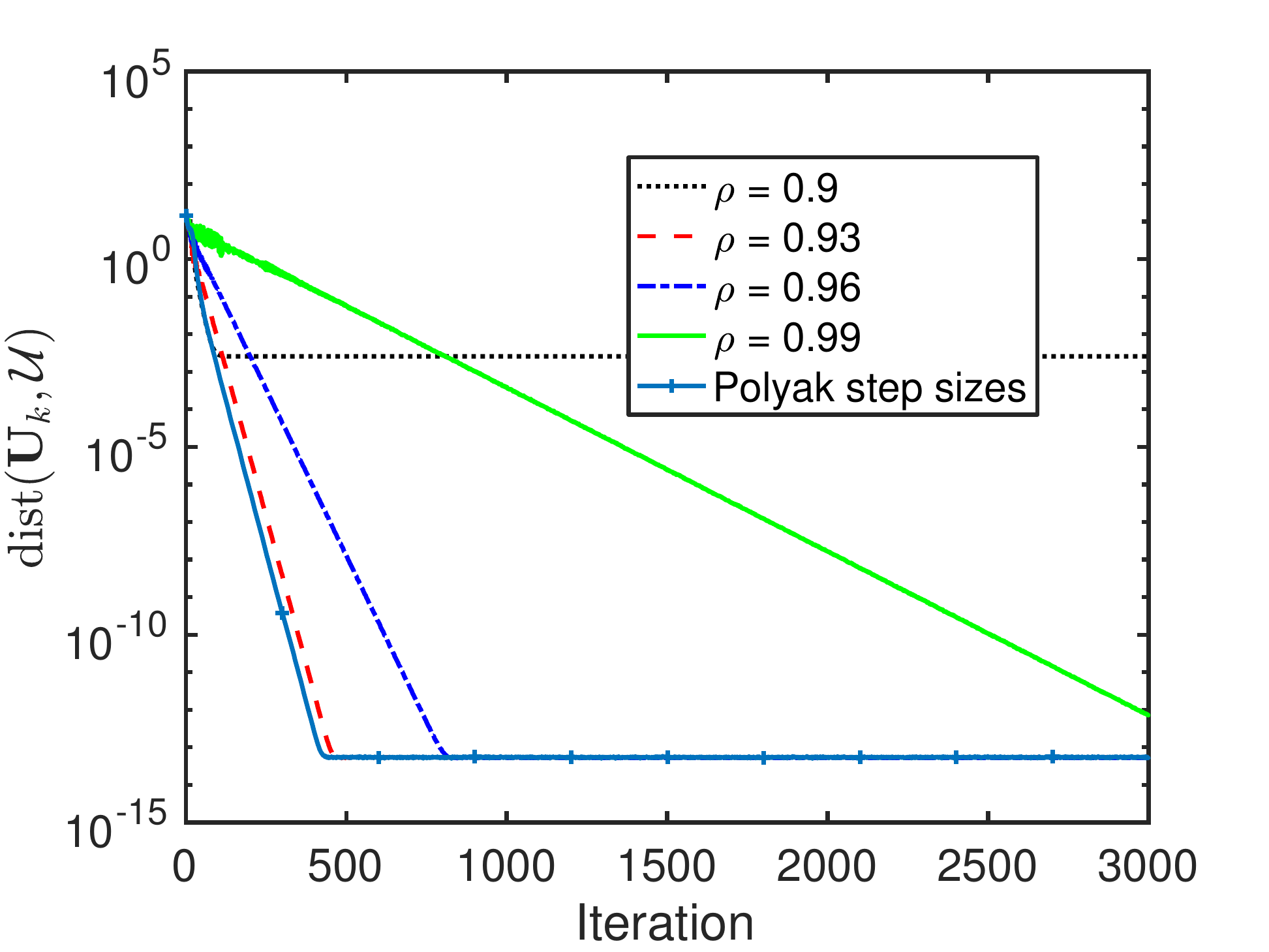}}
			\caption{Convergence of SubGM with geometrically diminishing ($\mu_k=\rho^k$, $\rho\in\{0.90,0.93,0.96,0.99\}$) and Polyak step sizes}
			\label{fig: rms-gd}
		\end{subfigure}\\
	\vfill
	    \begin{subfigure}{0.49\linewidth}
	    	\centerline{
	    		\includegraphics[width=1\linewidth]{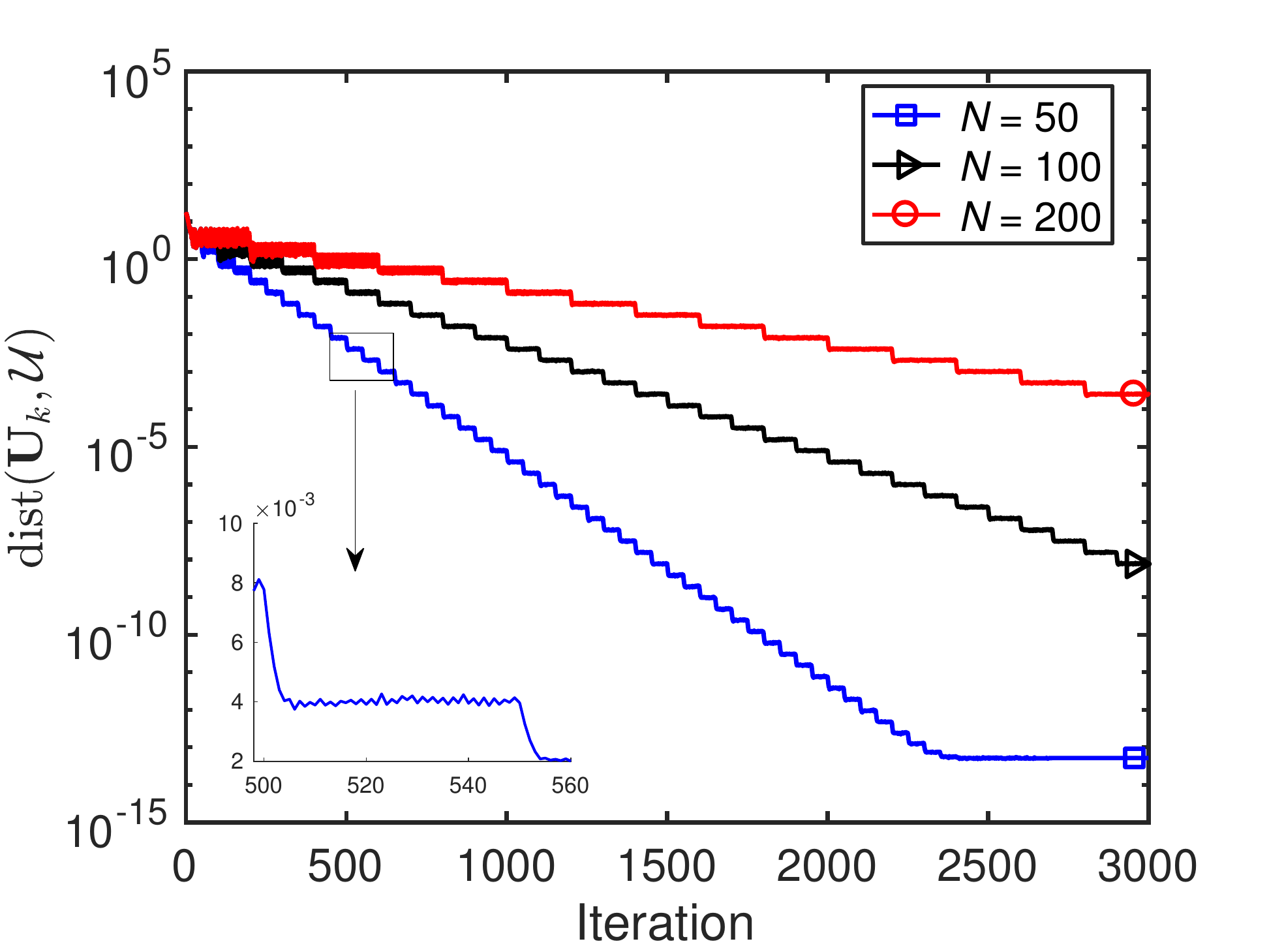}}
	    	\caption{Convergence of SubGM with piecewise geometrically diminishing ($\mu_k = \tfrac{1}{2^{\lfloor k/N \rfloor}}$, $N\in\{50,100,200\}$) step sizes}
	    	\label{fig: rms-pgd}
	    \end{subfigure}
	    \begin{subfigure}{0.49\linewidth}
	    	\centerline{
	    		\includegraphics[width=1\linewidth]{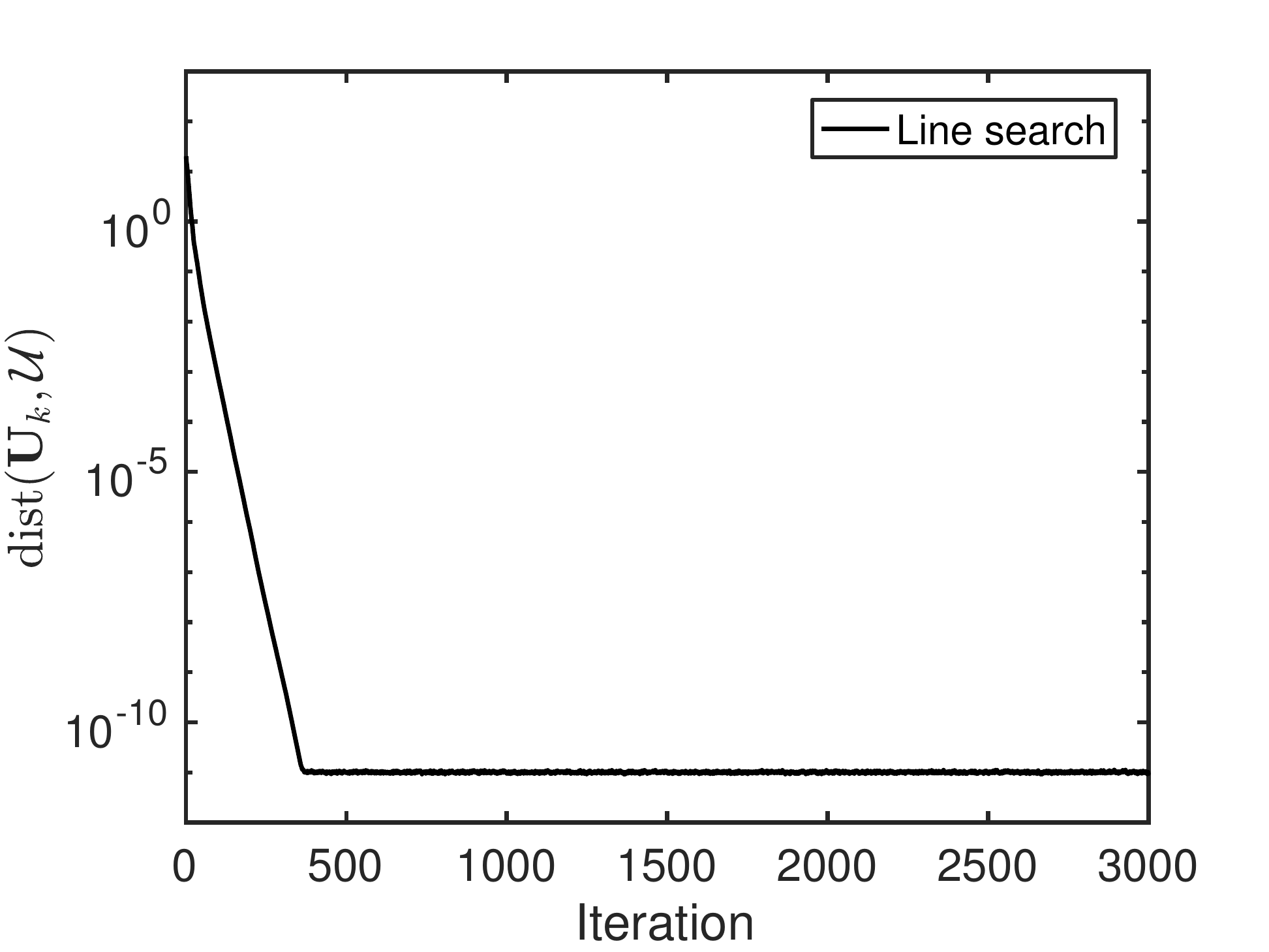}}
	    	\caption{Convergence of SubGM with modified backtracking line search ($\eta=10^{-3}$, $\rho=0.85$, $\mu_0=1$)}
	    	\label{fig: rms-line}
	    \end{subfigure}
	\caption{{Behavior of SubGM when applied to robust low-rank matrix recovery with $n=50$, $r=5$, $m = 5nr$, and $p = 0.3$.
}}\label{fig: rms}
\end{figure}

Next, we study the performance of the SubGM with geometrically diminishing step sizes by varying the outlier ratio $p$ and the number of measurements $m$. In these experiments we run the SubGM for $2\times 10^3$ iterations with initial step size $\mu_0 = 1$ and decay rate $\rho = 0.99$. We also conduct experiments on the median-truncated gradient descent (MTGD) with the setting used in~\cite{li2017nonconvex}. In particular, we initialize the MTGD with the truncated spectral method in \Cref{alg:intilization} and run it for $10^4$ iterations. For each pair of $p$ and $m$, 10 Monte Carlo trials are carried out, and for each trial we declare the recovery to be successful if the relative reconstruction error satisfies $\frac{\|\widehat \mX - \mX^\star\|_F}{\|\mX^\star\|_F} \leq 10^{-6},$ where $\widehat \mX$ is the reconstructed matrix.  \Cref{fig: rms phase transition} displays the phase transition of MTGD and SubGM using the average result of 10 independent trials. In this figure, white indicates successful recovery while black indicates failure.  It is of interest to observe that when the outlier ratio $p$ is small, both the SubGM and MTGD can exactly recover the underlying low-rank matrix $\mX^\star$ even with only $m = 2nr$ measurements. On the other hand, given sufficiently large number of measurements (say $m = 7nr$), the SubGM is able to exactly recover the ground-truth matrix even when half of the measurements are corrupted by outliers, while the MTGD fails in this case. In particular, by comparing \Cref{fig:phase MTGD} with \Cref{fig:phase SGM}, we observe that the SubGM is more robust to outliers than MTGD, especially in the case of high outlier ratio. We also observe from \Cref{fig: rms phase transition} that with more measurements, the robust low-rank matrix recovery formulation \eqref{eq:rms factorization} can tolerate not only more outliers but also a higher fraction of outliers. This provides further explanation to the observations made after the proof of \Cref{lem:sharpness rms}.

\begin{figure}[!htp]
		\begin{subfigure}{0.48\linewidth}
	\centerline{
		\includegraphics[width=2.5in]{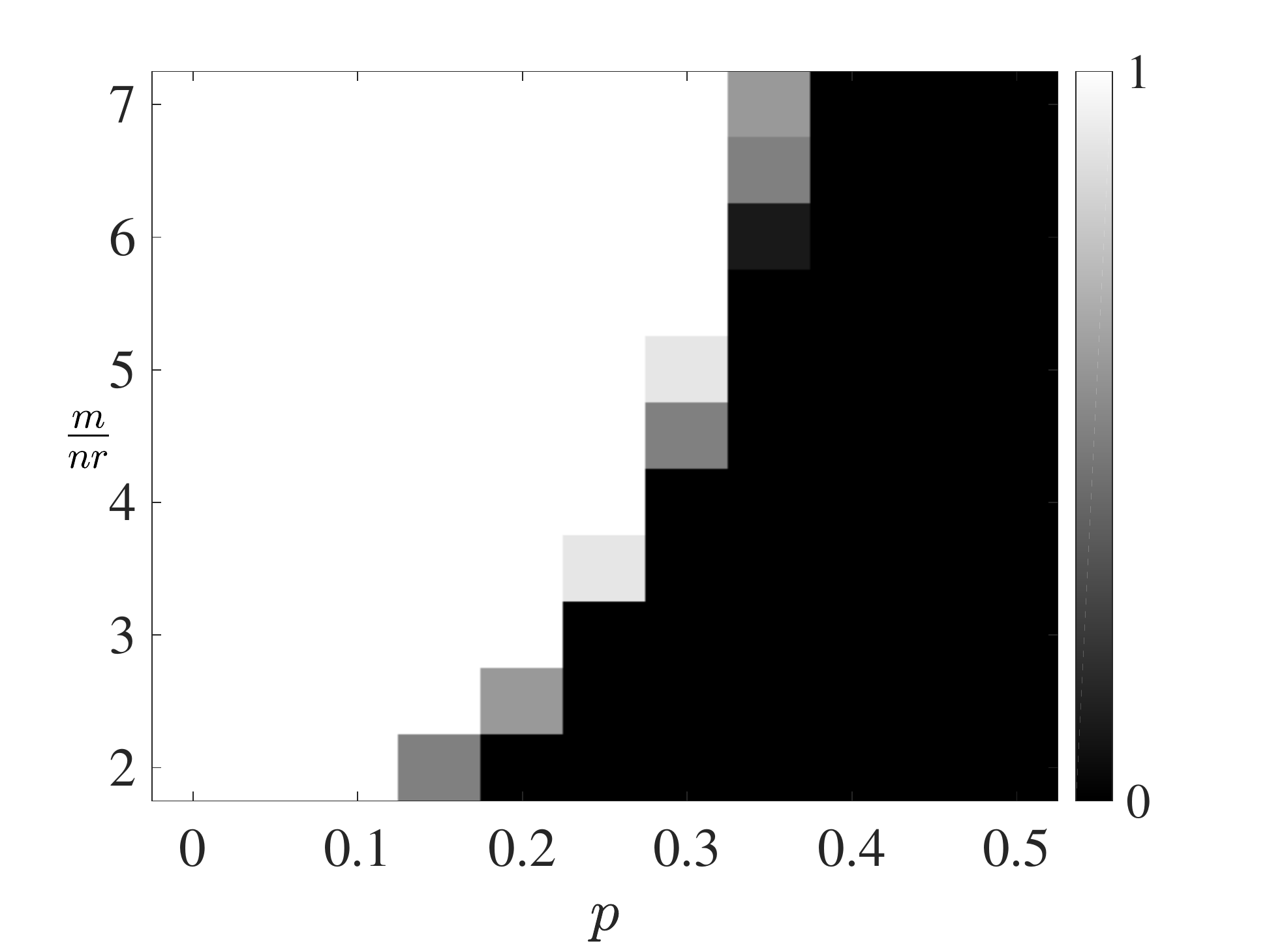}}
\caption{}\label{fig:phase MTGD}
	\end{subfigure}
		\begin{subfigure}{0.48\linewidth}
	\centerline{
		\includegraphics[width=2.5in]{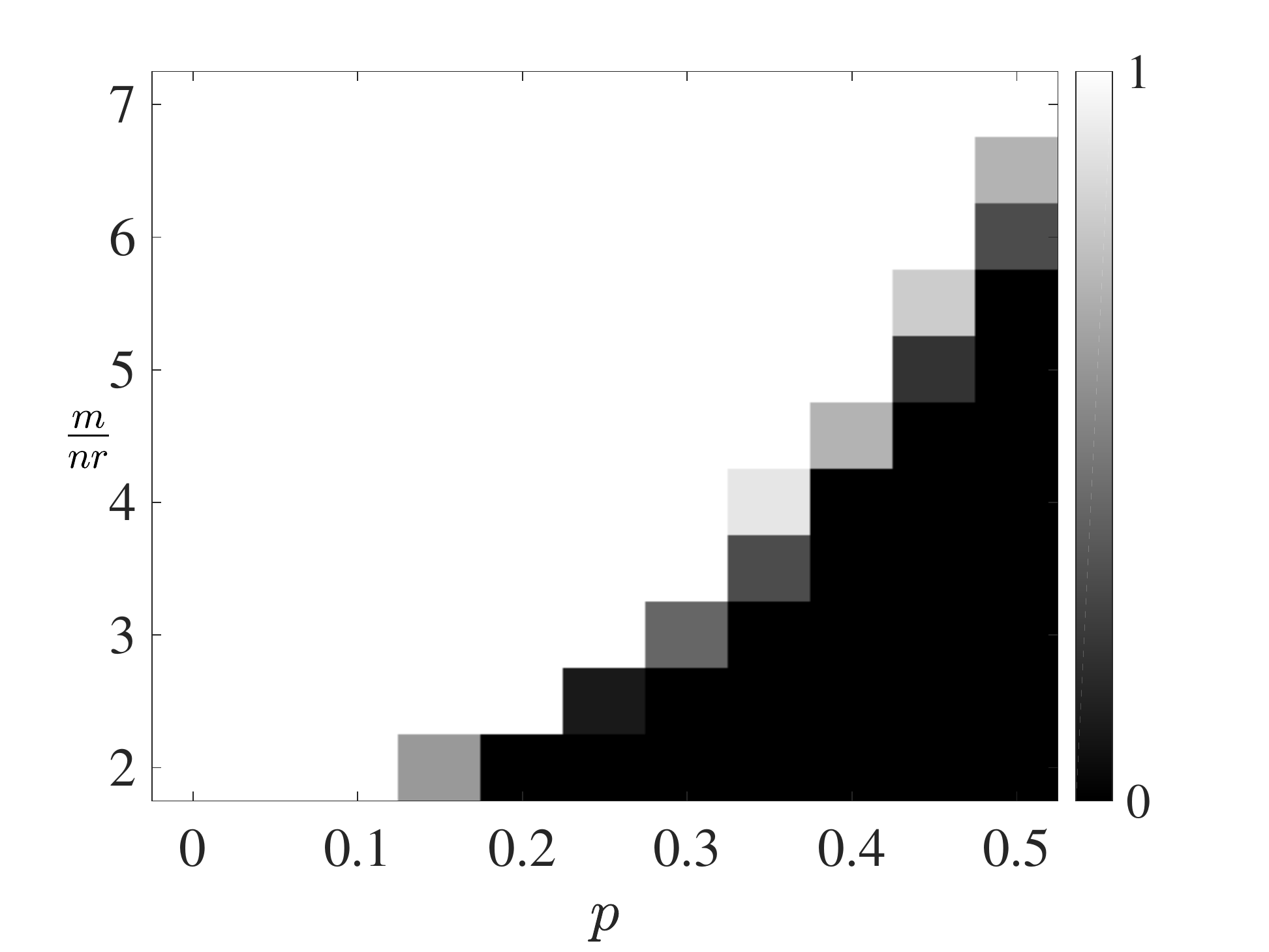}}
	\caption{}\label{fig:phase SGM}
\end{subfigure}
	\caption{Phase transition of robust low-rank matrix recovery using (a) median-truncated gradient descent (MTGD) \cite{li2017nonconvex} and (b) SubGM. Here, we fix $n = 50$, $r = 5$ and vary the outlier ratio $p$ from $0$ to $0.5$. In addition, we vary $m$ so that the ratio $\frac{m}{nr}$ varies from $2$ to $7$. Successful recovery is indicated by white and failure by black. Results are averaged over 10 independent trials.}\label{fig: rms phase transition}
\end{figure}

\section{Conclusion}
In this paper we gave a nonsmooth nonconvex formulation of the problem of recovering a rank-$r$ matrix $\mX^\star \in \R^{n_1\times n_2}$ from corrupted linear measurements. The formulation enforces the low-rank property of the solution by using a factored representation of the matrix variable and employs an $\ell_1$-loss function to robustify the solution against outliers. We showed that even when close to half of the measurements are arbitrarily corrupted, as long as certain measurement operators arising from the measurement model satisfy the $\ell_1/\ell_2$-RIP, the formulation will be sharp and weakly convex. Consequently, the ground-truth matrix can be exactly recovered from any of its global minimum. Moreover, when suitably initialized, the SubGM with geometrically diminishing step sizes will converge to the ground-truth matrix at a linear rate.

As the reader may note, our numerical experiments in \Cref{sec:experiments} suggest that the SubGM can efficiently find the underlying low-rank matrix even with a random initialization. This raises the question of whether there are spurious local minima in our formulation of the robust low-rank matrix recovery problem. Another question is whether the SubGM with a random initialization can escape saddle points and converge to a local minimum (which is also a global minimum if there is no spurious local minimum), just like the gradient descent for smooth problems~\cite{lee2016gradient}. We leave the study of these questions as future work.


\section{Acknowledgment} We thank the Associate Editor and two anonymous reviewers for their detailed and helpful comments.

\bibliographystyle{siamplain}
\bibliography{nonconvex}

\appendix
\section{Proof of \Cref{thm: L1-2 RIP for Gaussian}}
\label{sec:prf l1l2 rip}
\subsection{Preliminaries}
We say that a random variable $X$ is sub-Gaussian if
\[ 	\P{|X|>t}\leq \exp\left(1-\frac{t^2}{K_1^2}\right), \ \forall \ t\geq0 \]
for some constant $K_1>0$. This is equivalent to
\begin{equation} \label{eq:equivalence for subgaussian}
(\E[|X|^p])^{1/p}\leq K_2 \sqrt{p}, \ \forall \ p\geq 1
\end{equation}
for some constant $K_2>0$. The constants $K_1$ and $K_2$ differ from each other by at most an absolute constant factor; see~\cite[Lemma 5.5]{vershynin2010introduction}. The sub-Gaussian norm of a sub-Gaussian random variable $X$ is defined as
\[
\|X\|_{\psi_2} = \sup_{p\geq 1} \left\{ p^{-1/2} \E[|X|^p]^{1/p} \right\}.
\]
We then have the following Hoeffding-type inequalty:
\begin{lem}[{\cite[Proposition 5.10]{vershynin2010introduction}}]
	Let $X_{1},\ldots,X_m$ be independent sub-Gaussian random variables with $\E[X_i] = 0$ for $i=1,\ldots,m$ and $K = \max_{i\in\{1,\ldots,m\}} \|X_i\|_{\psi_2}$. Then, for any $t>0$, we have
	\e\label{eq:Hoeffding type inequality}
	\P{\frac{1}{m}\left|\sum_{i=1}^m X_i\right| >t }  \leq 2 \exp\left(-\frac{cmt^2}{K^2}\right)
	\ee
	for some constant $c>0$.
	\label{lem:Hoeffding type inequality}\end{lem}

We also need the following result on the covering number of the set of low-rank matrices:
\begin{lem}[{\cite[Lemma 3.1]{candes2011tight}}] Let $\setS_r = \{\mX\in\R^{n\times n}: \|\mX\|_F = 1, \rank(\mX)\leq r\}$. Then, there exists an $\epsilon$-net $\overline \setS_{r,\epsilon} \subset \setS_r$ with respect to the Frobenius norm (i.e., for any $\mX\in \setS_{r}$, there exists an $\overline \mX\in\overline \setS_{r,\epsilon}$ such that $\|\mX - \overline\mX\|_F\leq \epsilon$) satisfying
	\e
	|\overline \setS_{r,\epsilon}|\leq \left(\frac{9}{\epsilon}\right)^{(2n+1)r}.
	\label{eq:number in Sr}\ee  \label{lem:covering number}
\end{lem}

\subsection{Isometry Property of a Given Matrix}
\begin{lem}\label{lem: L1-2 RIP fxied X}
	Suppose that the matrices $\mA_1,\ldots,\mA_m \in \R^{n\times n}$ defining the linear measurement operator $\calA$ have \emph{i.i.d.} standard Gaussian entries. Then, for any $\mX \in \R^{n\times n}$ and $0<\delta<1$, there exists a constant $c_1>0$ such that with probability exceeding $1-2\exp(-c_1\delta^2 m)$, we have
	\e \label{eq:L1-2 RIX fixed X}
	\left( \sqrt{\frac{2}{\pi}} - \delta \right) \|\mX\|_F \leq	\frac{1}{m}  \|\calA( \mX)\|_1  \leq  \left( \sqrt{\frac{2}{\pi}} + \delta \right) \|\mX\|_F.
	\ee
\end{lem}

\begin{proof}[Proof of \Cref{lem: L1-2 RIP fxied X}]
	Since $\mA_i$ has \emph{i.i.d.} standard Gaussian entries, the random variable $\langle \mA_i,\mX \rangle$ is Gaussian with mean zero and variance $\|\mX\|_F^2$. It follows that
	\e
	\E[|\langle \mA_i,\mX\rangle| ] =  \sqrt{\frac{2}{\pi}}\|\mX\|_F, \ \ \E[\|\calA(\mX)\|_1 ] = m\sqrt{\frac{2}{\pi}}\|\mX\|_F.
	\label{eq:mean of <A,X>}\ee
	Now, let $Z_i = |\langle \mA_i,\mX\rangle| - \E[|\langle \mA_i,\mX\rangle| ]$, which satisfies $\E[Z_i] = 0$. We claim that $Z_i$ is a sub-Gaussian random variable. To establish the claim, it suffices to bound the sub-Gaussian norm of $Z_i$. Towards that end, we first observe that
	\[
	\P{|\langle \mA_i,\mX\rangle|>t}\leq 2\exp\left(-\frac{t^2}{2\|\mX\|_F^2}\right).
	\]
	Together with \eqref{eq:mean of <A,X>}, this implies that for any $t> \E[|\langle \mA_i,\mX\rangle|]$,
	\begin{align*}
	\P{|Z_i|>t} &= \P{|\langle \mA_i,\mX\rangle|>t + \E[|\langle \mA_i,\mX\rangle|]} + \P{|\langle \mA_i,\mX\rangle|<-t + \E[|\langle \mA_i,\mX\rangle|]} \\
	& \leq 2 \exp\left(-\frac{\left( t + \E[|\langle \mA_i,\mX\rangle|] \right)^2}{2\|\mX\|_F^2}\right) +  \P{|\langle \mA_i,\mX\rangle|<-t + \E[|\langle \mA_i,\mX\rangle|]}\\
	&  \leq 2 \exp\left(-\frac{\left( t + \E[|\langle \mA_i,\mX\rangle|] \right)^2}{2\|\mX\|_F^2}\right) \leq \exp\left( 1 - \frac{t^2}{\|\mX\|_F^2}\right),
	\end{align*}
	where the second inequality is from the fact that $\P{|\langle \mA_i,\mX\rangle|<-t + \E[|\langle \mA_i,\mX\rangle|]} = 0$ for all $t> \E[|\langle \mA_i,\mX\rangle|]$. Since $\exp\left( 1 - \frac{t^2}{\|\mX\|_F^2} \right) \geq 1$ for all $t\leq E[|\langle \mA_i,\mX\rangle|] = \sqrt{\frac{2}{\pi}}\|\mX\|_F$, we then have
	\[
	\P{|Z_i|>t}  \leq \exp\left(1 - \frac{t^2}{\|\mX\|_F^2}\right), \ \forall t\geq 0.
	\]
	This, together with \eqref{eq:equivalence for subgaussian}, implies that
	\[
	(\E[|Z_i|^p])^{1/p}\leq cp^{1/2}\|\mX\|_F, \ \forall \ p\geq 1,
	\]
	where $c>0$ is a constant. It follows that
	\e \label{eq:subexponential norm}
	\|Z_i\|_{\psi_2} \leq c\|\mX\|_F;
	\ee
	i.e., $Z_i$ is a sub-Gaussian random variable, as desired.
	
	Now, applying the Hoeffding-type inequality in \Cref{lem:Hoeffding type inequality} with $t = \delta \|\mX\|_F$ and $K = c\|\mX\|_F$ gives
	\[
	\P{\frac{1}{m}\left| \|\calA(\mX)\|_1 - \E[\|\calA(\mX)\|_1] \right| > \delta \|\mX\|_F} \leq 2\exp(-c_1m\delta^2)
	\]
	for some constant $c_1>0$. Using~\eqref{eq:mean of <A,X>}, we conclude that~\eqref{eq:L1-2 RIX fixed X} holds with probability at least $1 - 2\exp(-c_1m \delta^2)$. This completes the proof.
\end{proof}

\subsection{Proof of \Cref{thm: L1-2 RIP for Gaussian}}
We now utilize an $\epsilon$-net argument to show that \eqref{eq:L1-2 RIX fixed X} holds for all rank-$r$ matrices with high probability as long as $m\gtrsim nr$.  Since the inequality~\eqref{eq:L1-2 RIX fixed X} is scale invariant, without loss of generality, we may assume that $\|\mX\|_F = 1$ and focus on the set $\setS_{r}$ defined in \Cref{lem:covering number}.
\begin{proof}[Proof of \Cref{thm: L1-2 RIP for Gaussian}]
	We begin by showing that~\eqref{eq:L1-2 RIX fixed X} holds for all $\mX \in \overline \setS_{r,\epsilon}$ with high probability. Indeed, upon setting $\epsilon = \frac{\delta\sqrt{\pi}}{16}$ in \eqref{eq:number in Sr} and utilizing a union bound together with \cref{lem: L1-2 RIP fxied X}, we have
	\e\begin{split}
		&\P{\max_{\overline\mX\in\overline \setS_{r,\epsilon}}  \frac{1}{m}\left|\|\calA(\overline\mX)\|_1 - m\sqrt{\frac{2}{\pi}}\|\overline\mX\|_F \right| \geq \frac{\delta}{2}} \leq 2|\overline\setS_{r,\epsilon}| \exp(-c_1m\delta^2)\\
		&\leq 2 \left(\frac{9}{\epsilon}\right)^{(2n+1)r} \exp(-c_1 m \delta^2) \leq \exp(-c_2m\delta^2)
	\end{split}\label{eq:net-conc}
	\ee
	whenever $m\gtrsim nr$.
	
	Next, we show that~\eqref{eq:L1-2 RIX fixed X} holds for all $\mX \in \setS_r$. Towards that end, set
	\e
	\kappa_r = \frac{1}{m} \sup_{\mX\in\setS_r} \left\|\calA(\mX)\right\|_1
	\label{eq:kappa r}\ee
	and let $\mX\in \setS_r$ be \emph{arbitrary}. Then, there exists an $\overline \mX\in \overline \setS_{r,\epsilon}$ such that $\|\mX - \overline\mX\|_F \leq \epsilon$. It follows from~\eqref{eq:net-conc} that with high probability,
	\e\begin{split}
		\frac{1}{m}\|\calA(\mX)\|_1 &= \frac{1}{m}\|\calA(\mX - \overline\mX) + \calA( \overline\mX)\|_1 \leq \frac{1}{m}\|\calA(\mX - \overline\mX) \|_1 + \frac{1}{m}\| \calA(\overline\mX)\|_1\\
		&\leq \frac{1}{m}\|\calA(\mX - \overline\mX) \|_1 + \sqrt{\frac{2}{\pi}} + \frac{\delta}{2}.
	\end{split}\label{eq:bound covering 1}\ee
	Noting that $\mX - \overline\mX$ has rank at most $2r$, we can decompose it as $\mX - \overline\mX = \mDelta_1 + \mDelta_2$, where $\langle \mDelta_1, \mDelta_2\rangle = 0$ and $\rank(\mDelta_1),\rank(\mDelta_2)\leq r$ (this follows essentially from the SVD). Hence, we can compute
	\e\begin{split}
		& \frac{1}{m} \|\calA(\mX - \overline\mX) \|_1 \leq \frac{1}{m}[\|\calA(\mDelta_1)\|_1 + \|\calA(\mDelta_2)\|_1] \\
		&= \frac{1}{m}[\|\mDelta_1\|_F\|\calA(\mDelta_1/\|\mDelta_1\|_F)\|_1 + \|\mDelta_2\|_F\|\calA(\mDelta_2/\|\mDelta_2\|_F)\|_1]\\ & \leq \kappa_r (\|\mDelta_1\|_F + \|\mDelta_2\|_F) \leq \sqrt{2} \kappa_r \epsilon,
	\end{split}\nonumber\ee
	where the last inequality is due to $\|\mDelta_1\|_F^2 + \|\mDelta_2\|_F^2 = \|\mX - \overline\mX\|_F^2 \leq \epsilon^2$. This, together with \eqref{eq:bound covering 1}, gives
	\begin{equation} \label{eq:ub}
	\frac{1}{m}\|\calA(\mX)\|_1 \leq \sqrt{\frac{2}{\pi}} + \frac{\delta}{2} + \sqrt{2} \kappa_r \epsilon.
	\end{equation}
	In particular, using the definition of $\kappa_r$ in~\eqref{eq:kappa r}, we obtain
	\[
	\kappa_r \leq \sqrt{\frac{2}{\pi}} + \frac{\delta}{2} + \sqrt{2}\kappa_r\epsilon,
	\]
	or equivalently,
	\[ \kappa_r \leq \frac{\sqrt{{2}/{\pi}} + \delta/2}{ 1- \sqrt{2}\epsilon}. \]
	Plugging in our choice of $\epsilon$ yields 
$\sqrt{2}\kappa_r\epsilon \le  \frac{\delta}{2}$. This, together with~\eqref{eq:ub} and the fact that $\|\mX\|_F=1$, implies
	\[ \frac{1}{m}\|\calA(\mX)\|_1 \le \left( \sqrt{\frac{2}{\pi}} + \delta \right) \|\mX\|_F. \]
	Similarly, using~\eqref{eq:net-conc}, we have
	\begin{align*}
	& \frac{1}{m}\|\calA(\mX)\|_1 \geq  \frac{1}{m}\| \calA(\overline\mX)\|_1 -  \frac{1}{m}\|\calA(\mX - \overline\mX) \|_1 \\
	& \ge  \sqrt{\frac{2}{\pi}}  - \frac{\delta}{2} -  \frac{1}{m}\|\calA(\mX - \overline\mX) \|_1 \\
	&\ge \sqrt{\frac{2}{\pi}}  - \frac{\delta}{2} - \sqrt{2}\kappa_r\epsilon \ge \sqrt{\frac{2}{\pi}}  - \delta \\
	&= \left( \sqrt{\frac{2}{\pi}}  - \delta \right) \|\mX\|_F
	\end{align*}
	with high probability. This completes the proof.
\end{proof}

\end{document}